\documentclass[12pt]{iopart}
\usepackage{iopams}

\usepackage{longtable}
\usepackage{lscape}  

\newtheorem{lemma}{Lemma}
\newtheorem{theorem}{Theorem}
\newtheorem{definition}{Definition}

\newcommand{\diagonal}[1]{\mathrm{diag}\left[#1\right]} 
\newcommand{\RMset}[1]{\mathbf{M}_{#1}} 
\newcommand{\Sym}[1]{\mathbf{Sym}_{#1}} 

\begin{document}

\title{One-dimensional subspaces of the $SL(n,\mathbb{R})$ Chiral Equations}

\author{I. A. Sarmiento-Alvarado$^1$, Petra Wiederhold$^2$ and Tonatiuh Matos$^1$}
\address{$^1$ Departamento de F\'{\i}sica, Centro de Investigaci\'on y de Estudios Avanzados del IPN, Av. I.P.N. 2508, San Pedro Zacatenco, M\'exico 07360, CDMX.}
\address{$^2$ Departamento de Control Autom\'atico, Centro de Investigaci\'on y de Estudios Avanzados del IPN, Av. I.P.N. 2508, San Pedro Zacatenco, M\'exico 07360, CDMX.}
\ead{ignacio.sarmiento@cinvestav.mx , tonatiuh.matos@cinvestav.mx , petra.wiederhold@cinvestav.mx}

\date{\today}

\vspace{10pt}

\begin{abstract}
In this work we find solutions of the ($n+2$)-dimensional Einstein Field Equations (EFE) with $n$ commuting Killing vectors in vacuum. In the presence of $n$ Killing vectors, the EFE can be separated into blocks of equations. The main part can be summarized in the chiral equation
$\  (\alpha g_{, \bar{z}} g^{-1})_{, z} + \  (\alpha g_{, z} g^{-1})_{, \bar{z}} = 0$ with $ g\in SL(n,\mathbb{R})$. The other block reduces to the differential equation $\, (\ln f \alpha ^{1-1/n})_{, z} = 1/2 \, \alpha \tr( g_{, z} g^{-1})^2$ and its complex conjugate.
We use the ansatz $g = g(\xi ) $, where $\xi $ satisfies a generalized Laplace equation, so the chiral equation reduces to a matrix equation that can be solved using algebraic methods, turning the problem of obtaining exact solutions for these complicated differential equations into an algebraic problem. The different EFE solutions can be chosen with desired physical properties in a simple way.
\end{abstract}

\vspace{2pc}
\vspace{2pc}
%
\maketitle
%
%

\section{Introduction}

The Einstein Field Equations (EFE) are one of the most interesting field equations in physics and from a mathematical point of view, the search for methods to obtain solutions has led to a large number of mathematical results. The first exact solution was obtained by Karl Schwarzschild in 1916, but his study has been a long debate over the meaning of the solution. We now know that the Schwarzschild solution represents a static black hole. His generalization for a stationary solution had to wait more than 40 years to be found by Roy Kerr. These solutions have been the cornerstone of the theory of general relativity and its interpretation and represent a stationary black hole.
After the finding of Kerr's solution, the exact solutions area of the EFE has been very active, see for example \cite{exact03}. Several mathematical methods have been developed with great success to find exact solutions of EFEs. One of the most successful has been the method of subspaces and subgroups, which is capable of generating exact solutions on demand. It is possible to decide the exact physical content of the solution from the beginning. That is why in this work we will adopt this solution method.

On the other hand, interest in higher dimensional theories began in 1919 with Theodor Kaluza's proposal for a five-dimensional space-time that unified gravitation with electromagnetism. Kaluza proposes that the metric of a five-dimensional spacetime can be separated as $g^5_{\mu\nu}=g^4_{\mu\nu}+I^2A_\mu A_\nu$,  for $\mu,\nu=1,\cdots,4$, $g^5_{5\mu}=IA_\mu$ and $g^5_{55}=I^2$, where $A_\mu$ is the tetraelectromagnetic potential and $I$ is related to a scalar field, called dilaton, see for example \cite{Matos:1992qx}. This theory has evolved to the unification of all interactions; electromagnetic, strong and weak interactions with gravity. However, the theory includes quantum interactions, but it is not renormalizable, nor is it quantizable.
So people propose string and superstring theory to have a quantizable renormalizable higher dimensional theory, see for example \cite{SuperS85}. The price they have to pay is that the extra dimensions must be singular. In this paper we propose that the extra dimensions form an $n$-dimensional space with $n-2$ Killing vectors that may be singular and interesting enough to be studied.

In this work we pretend to find exact solutions of the EFE from the mathematical point of view, using the method of space and subgroups which seems very successful to obtain a great amount of exact solutions.

We start with an ($n+2$)-dimensional space and are interested in 4-dimensional spacetimes that are stationary and axially symmetric. This means that the 4-dimensional spacetime contains two commuting Killing vectors whose extra dimensional space is ($n-2$)-dimensional with $n-2$ commuting Killing vectors, so that the n-dimensional space contains $n$ commuting Killing vectors. Thus, in this case we can work in a coordinate system where the metric depends only on two variables $x^1$ and $x^2$, so that the metric tensor has the form
\begin{equation}
\label{eq:metric}
\hat{g} = f \  ( \rmd x ^1\otimes \rmd x ^1 +  \rmd x ^2 \otimes \rmd x ^2 ) + g_{\mu _\nu } \rmd x ^\mu\otimes \rmd x ^\nu
\end{equation}
where the components of $\hat{g}$, $f$ and $g$, for $\mu, \nu = 3, \ldots , n+2$, depend on two variables $x^1$ and $x^2$. In the following, we will denote the uppercase indices as $A,B=1,\ldots, n+2$ and the Greek indices as $\mu, \nu = 3,\ldots , n+2$.

Throughout this paper, the set of matrices of size $m\times n$ with entries in $\mathbb{R}$ is denoted by $ \RMset{m \times n} $, we write $\mathbf{M}_{m}$ if $m = n$.
The identity matrix and the zero matrix are denoted by $I_n$ and $0_n$, respectively, and, $\Sym{n}$ is the subset of symmetric matrices in $\mathbf{M}_{n}$.

This work is organized as follows. In Section \ref{sec:Field_equations} we follow \cite{Matos89} to write the main field equations, obtaining the Ricci tensor for an ($n+2$)-dimensional space with $n$ commutative Killing vectors. Section \ref{sec:subspace} presents the algebraic basis of the matrices used in this work. In Section \ref{sec:jordan} we express these matrices by their Jordan normal form in order to solve the final algebraic equation. In Section \ref{sec:example},  using the Jordan form of matrices we obtain the solutions of the algebraic equations. Finally, Section \ref{sec:conclusions} contains some conclusions.

\section{Field equations}\label{sec:Field_equations}

In this section we derive the main field equations for an ($n+2$)-dimensional Riemannian space with $n$ commuting Killing vectors. The metric components depend only on the coordinates $\hat{g}_{A B}=\hat{g}_{A B}(x^1,x^2)$. In this case, the Christoffel symbols are given as
\begin{equation}
    \Gamma^C_{A B} \equiv
    \frac{1}{2} \hat{g}^{C D} \  (\hat{g}_{D A , B} + \hat{g}_{D B , A} - \hat{g}_{A B , D}  )\  .
\end{equation}
For the metric (\ref{eq:metric}) we have
\begin{equation}
\begin{array}{cccccccc}
\fl \Gamma^1_{1 1}
&   = \frac{1}{2} ( \ln f )_{, 1}
&  ,\  \Gamma^1_{1 2}
&   = \frac{1}{2} ( \ln f )_{, 2}
&  ,\  \Gamma^1_{2 2}
&   = - \frac{1}{2} ( \ln f )_{, 1}
&  ,\  \Gamma^{\mu } _{\nu i}
&   = \frac{1}{2} g^{\mu \omega } g_{\omega _\nu , i}  ,
\\  \fl \Gamma^2_{2 2}
&   = \frac{1}{2} ( \ln f )_{, 2}
&  ,\  \Gamma^2_{1 2}
&   = \frac{1}{2} ( \ln f )_{, 1}
&  ,\  \Gamma^2_{1 1}
&   = - \frac{1}{2} ( \ln f )_{, 2}
&  ,\  \Gamma^i _{\mu \nu}
&   = - \frac{1}{2} g_{\mu \nu } ^{, i}  ,
\end{array}
\end{equation}
the remaining components are zero.


In order to compute the Ricci tensor with our metric,
\begin{equation}
    R_{A B}
    = \Gamma^C_{A B , C} - \Gamma^C_{A C , B} + \Gamma^C_{D C} \Gamma^D_{A B}
    - \Gamma^C_{D B} \Gamma^D_{A C}\  ,
\end{equation}
it is convenient to use the variables $z = x _1 + i x _2 $ and its complex conjugate, ${\bar z}$.
Hence, the non-zero components of the Ricci tensor are as follows:
\begin{eqnarray}
\fl    R_{1 1}
   = -2 (\ln f \alpha )_{, z {\bar z}}
    - \frac{1}{2} g^{\mu \alpha} g_{\alpha \nu , z} g^{\nu \beta} g_{\beta \mu , {\bar z}}
    - (\ln \alpha )_{, z z}
    + (\ln \alpha )_{, z} (\ln f)_{, z} \\
    - \frac{1}{4} g^{\mu \alpha} g_{\alpha \nu , z } g^{\nu \beta} g_{\beta \mu , z}
    - (\ln \alpha )_{, {\bar z} {\bar z}}
    + (\ln \alpha )_{, {\bar z}} (\ln f )_{, {\bar z}}
    - \frac{1}{4} g^{\mu \alpha} g_{\alpha \nu , {\bar z}} g^{\nu \beta} g_{\beta \mu , {\bar z}}
\\ \fl  R_{2 2}
   = -2 (\ln f \alpha )_{, z {\bar z}}
    - \frac{1}{2} g^{\mu \alpha} g_{\alpha \nu , z} g^{\nu \beta} g_{\beta \mu , {\bar z}}
    + (\ln \alpha )_{, z z}
    - (\ln \alpha )_{, z} (\ln f)_{, z} \\
    + \frac{1}{4} g^{\mu \alpha} g_{\alpha \nu , z} g^{\nu \beta} g_{\beta \mu , z}
    + (\ln \alpha )_{, {\bar z} {\bar z}}
    - (\ln \alpha )_{, {\bar z}} (\ln f )_{, {\bar z}}
    + \frac{1}{4} g^{\mu \alpha} g_{\alpha \nu , {\bar z} } g^{\nu \beta} g_{\beta \mu , {\bar z}}
\\  \fl R_{1 2}
   = i \  [
        (\ln \alpha )_{, {\bar z} {\bar z}}
        - (\ln \alpha )_{, {\bar z}} (\ln f)_{, {\bar z}}
        + \frac{1}{4} g_{\alpha \beta , {\bar z}} g^{\beta \gamma } g_{\gamma \delta , {\bar z}} g^{\delta \alpha }
        - (\ln \alpha )_{, z z}
        + (\ln \alpha )_{, z} (\ln f )_{, z} \\
        - \frac{1}{4} g_{\alpha \beta , z} g^{\beta \gamma } g_{\gamma \delta , z} g^{\delta \alpha }
    ]
\\  \fl R_{\mu } ^{\nu}
   = - \frac{1}{f\alpha } \  [
        \  (\alpha g_{\mu \omega , {\bar z}} g^{\omega \nu } )_{, z}
        + \  ( \alpha g_{\mu \omega , z} g^{\omega \nu } )_{, {\bar z}}  ]
\end{eqnarray}
where $\det g_{\mu \nu } = - \alpha ^2$.


We will use matrix notation, let us define the matrix $g$ from the components of the metric tensor
$g_{\mu \nu}$ as follows:
\begin{equation}
    (g)_{\mu \nu} = g_{\mu \nu }\  .
\end{equation}
Note that the matrix $g$ is real and symmetric, that is, denoting by $T$ transpose of a matrix,
\begin{eqnarray}
\label{eq:properties_g det}
\det g   &=& -\alpha ^2
\\\label{eq:properties_g real}
{\bar{g}}  &=& g
\\\label{eq:properties_g symmetric}
g^T  &=& g
\end{eqnarray}

The vacuum Einstein equations are given by
\begin{equation}
    R_{A B} = 0 \  .
\end{equation}
From $R_{\mu } ^{\nu } = 0 $ we obtain the chiral equations
\begin{equation}
\label{eq:chiral eq}
    \  ( \alpha g_{, {\bar z}} g^{-1} ) _{, z}
    + \  ( \alpha g_{, z} g^{-1} )_{, {\bar z}} = 0 \  .
\end{equation}
Its trace gives a differential equation for $\alpha $:
\begin{equation}
\label{eq:trace chiral eq}
    \alpha_{, z {\bar z}} = 0 \  .
\end{equation}
From now on, the index $Z$ will take the values $z$ and $\bar z$.
Using $R_{1 1} - R_{2 2} \pm 2i R_{1 2} = 0 $ we find
\begin{equation}
\label{eq:diff eq f}
(\ln f \alpha )_{, Z} = \frac{\alpha_{, Z Z}}{\alpha_{, Z}}
    + \frac{\tr( g_{, Z} g^{-1} )^2}{4(\ln \alpha )_{, _Z}}\  .
\end{equation}
Both Equations (for $z$ and ${\bar z}$) satisfy
\begin{equation}
    (\ln f \alpha )_{, z {\bar z}}
    = - \frac{1}{4} \tr ( g_{, z} g^{-1}  g_{, {\bar z}} g^{-1} ) \  .
\end{equation}
Using the transformation
\begin{equation}
\label{eq:g to det g}
     g \to -\alpha ^{-2/n} g
\end{equation}
we normalize $g$, i.e., $\det g = ( -1 ) ^{n+1}$.
Therefore, $g$ is a symmetric matrix in $SL(n,\mathbb{R})$.

The chiral equation (\ref{eq:chiral eq}) does not change under the transformation (\ref{eq:g to det g}), whereas Equation (\ref{eq:diff eq f}) takes the form
\begin{equation}
\label{eq:eq diff eq f g prime}
    ( \ln f \alpha ^{1-1/n} )_{, Z}
    = \frac{ \alpha _{, Z Z} }{\alpha _{, Z} }
    + \frac{\tr ( g_{, Z} g^{-1} )^2}{4(\ln \alpha )_{, Z} } \  .
\end{equation}
The chiral equation (\ref{eq:chiral eq}) is invariant under transformations
\begin{equation}
\label{eq:transformation g}
    g \to C g C ^T
\end{equation}
where $ C \in  SL(n,\mathbb{R}) $ is a constant matrix.
The general solution of the differential equation (\ref{eq:trace chiral eq}) for $\alpha$  is given as
\begin{equation}
    \alpha ( z , \bar z ) = \alpha _z ( z ) + \alpha _{\bar z} ( \bar z )
\end{equation}
where $ \alpha _z $ and $ \alpha _{\bar z} $ are arbitrary functions.
Chosing Weyl coordinates, i.e.,
\begin{equation}
    \alpha = \frac{ z + \bar z }{2}  \  ,
\end{equation}
Equations (\ref{eq:eq diff eq f g prime}) are reduced to
\begin{equation}
\label{diff eqs f}
    ( \ln f \alpha ^{1-1/n} )_{, Z}
    = \frac{1}{2} \alpha \tr ( g_{, Z} g^{-1} )^2 \  .
\end{equation}

The next sections will introduce important quantities to transform the differential equations (\ref{eq:chiral eq}).


\section{One-dimensional subspaces}\label{sec:One}

Suppose that $g$ depends on parameters $\xi $ which are arbitrary functions of the variables $z$ and $\bar z$.
Then, the chiral equation (\ref{eq:chiral eq}) changes to
\begin{equation}
    2\alpha \  (g_{, \xi } g^{-1})_{, \xi } \xi_{, z} \xi_{, {\bar z}}
    + g_{, \xi } g^{-1} \  ((\alpha \xi_{, z})_{, {\bar z}}
        + (\alpha \xi_{, {\bar z}} )_{, z} ) = 0 \  .
\end{equation}
Now we assume that the parameter $\xi $ satisfies the Laplace equation
\begin{equation}
\label{laplace eq}
    ( \alpha \xi_{, z} )_{, {\bar z}} + ( \alpha \xi_{, {\bar z}} )_{, z} = 0  \  ,
\end{equation}
then $g_{, \xi } g^{-1} = A$ is a constant matrix.
Note that each new solution of the Laplace equation gives another solution for $g$.
From the properties of the matrix $g$ we obtain
\begin{eqnarray}
\label{eq:properties_A real}
{\bar A} &=& A
\\
\label{eq:properties_A trace}
\tr A &=& 0
\\
\label{eq:properties_A ia}
A g &=& g A^T
\end{eqnarray}
Equations (\ref{eq:properties_A real}) and (\ref{eq:properties_A trace}) imply that $A$ belongs to the Lie algebra $\mathfrak{sl} (n,\mathbb{R} )$, the Lie algebra corresponding to the group $SL(n,\mathbb{R})$.
The matrix $A$ varies as
\begin{equation}
\label{eq:transformation A}
    A \to C A C ^{-1}
\end{equation}
under the transformation (\ref{eq:transformation g}).
The relation (\ref{eq:transformation A}) separates the set of matrices $A$ into equivalence classes.
We will work with a representative matrix of each class.

\section{The subspace $ \mathcal{I} ( A ) $}\label{sec:subspace}    

It is possible to find the general form of $g$ given $A$ if we consider the property (\ref{eq:properties_g symmetric}), together with the intertwining relation (\ref{eq:properties_A ia}) satisfied for the matrix $A$.
To do so, let us define the following set.

\begin{definition}
\label{def:subspace_ia}
For any non-zero matrix $ A \in \mathbf{M}_{n}$, define the set $\mathcal{I}( A )$ as
\begin{equation}
\mathcal{I}( A ) = \  \{ g \in \mathbf{Sym}_n : A g = g A ^T \}\  .
\end{equation}
\end{definition}

Observe that $ g \in \mathcal{I}(A) $.
Thus, the problem of finding the form of $g$ has been transformed into a linear algebra problem.
First, let us derive the following useful properties.

\begin{theorem}
   For any non-zero matrix $ A \in \mathbf{M}_{n}$,
    $\mathcal{I}( A )$ is a subspace of the vector space $\mathbf{M}_{n}$.
\end{theorem}

\noindent \textit{Proof:}
Let $\alpha \in \mathbb{R}  $ and let $X, Y \in \mathcal{I}(A)$.
We have $ ( \alpha X ) ^T = \alpha X ^T = \alpha X$ and $ ( X + Y ) ^T = X ^T + Y ^T = X + Y $.
Then, $A ( \alpha X ) = \alpha ( A X ) = \alpha ( X A ^T ) = ( \alpha X ) A ^T $ and $ A ( X + Y ) = A X + A Y = X A ^T + Y A ^T = ( X + Y ) A ^T $, so that
$\alpha X \in \mathcal{I}(A)$ and $X + Y \in \mathcal{I}(A)$.

\rightline{$\square$}





\begin{definition}
    For any non-zero matrix $ A \in \mathbf{M}_{n}$ and $ \xi \in \mathbb{R} $, define
    \begin{equation}
        e^{ \xi A } = \sum _{ k = 0 } ^\infty \frac{ \xi ^k }{ k ! } A ^k \  .
    \end{equation}
\end{definition}
For more information on the exponential matrix, see, for example, \cite{hall2015lie, L_Tu}. The following lemmas are corollaries of the above.

\begin{lemma}
\label{lemma: eA X = X eA ^T}
        Let $ A, g \in \mathbf{M}_{n}$ be non-zero matrices and $ \xi \in \mathbb{R}  $. Then
        $ e ^{ \xi A } g = g e ^{ \xi A ^T } $ if and only if $ A g = g A ^T $.
\end{lemma}

\noindent \textit{Proof:}
    We define the matrix function $ F ( \xi ) = e ^{ \xi A } g e ^{ -\xi A ^T } $.
    Its derivative is $ F' ( \xi ) = e ^{ \xi A } \  ( A g - g A ^T ) e ^{ -\xi A ^T } $.
    If $ g \in \mathcal{I} ( A ) $, then $ F' ( \xi ) = 0 $, so that $ F ( \xi ) = F ( 0 ) $.
    Therefore, $ e ^{ \xi A } g = g e ^{ \xi A ^T } $.
    Now, if $ e ^{ \xi A } g = g e ^{ \xi A ^T } $, then $ F ( \xi ) = F ( 0 ) $.
    Its derivative at $ \xi = 0 $ gives $ A g = g A ^T $.

\rightline{$\square$}







It is convenient to reduce the matrices we work with to simple matrices using the equivalence relation (\ref{eq:transformation g}). In particular, to facilitate the computation of the matrix exponentials, we will use the Jordan matrices introduced in the next section.

\section{Jordan matrices}\label{sec:jordan}

The invariance (\ref{eq:transformation g}) allows to use normal forms for the matrix $A$ which then is used to determine the matrix $g$. In this work we choose the real Jordan form of a matrix, because of its simplicity, and, since in this representation the matrix $A$ is always real even if it has complex conjugate eigenvalues. For an example of using the natural normal form of matrices instead of the Jordan form, see \cite{Matos:1992xn}, where the group $SL(3,\mathbb{R} )$ was discussed. Here we are going to focus on the group $SL(5,\mathbb{R} )$ in its Jordan representations.
For more information on the Jordan form, see, for example, \cite{Gantmacher59, horn_johnson_1985, Matos}.

It is well-known that any real square matrix may have real and complex eigenvalues, where for each complex eigenvalue $\alpha +\beta i$, also its complex conjugate $\alpha -\beta i$ is an eigenvalue. To avoid to include the complex values explicitly in the Jordan matrix, it is possible to include each such pair $\alpha \pm \beta i$ as represented by a real 2x2-matrix
\begin{equation}
\Lambda = \left[
\begin{array}{cc} \alpha & -\beta \\\beta & \alpha\end{array}\right] \  .
\end{equation}
Therefore, we will consider Jordan blocks of two kinds, one for the real eigenvalues and another type for the pairs of complex conjugate eigenvalues.
Furthermore, it is convenient for our work to represent the Jordan matrices as decomposed into blocks which make visible the type of eigenvalues.
In consequence, we introduce several types of Jordan blocks and matrices, more general as the standard notions from the common literature, as follows.


\begin{definition}
    For $ \lambda \in \mathbb{R} $, a \textbf{Jordan cell} $ J _n (\lambda) \in \mathbf{M}_{n}$ is an upper triangular matrix of the form
    \begin{equation}
        J _n (\lambda) = \left[ \begin{array}{ccccc}
            \lambda &   1       &   0       &   \cdots  &   0   \\
                    &   \lambda &   1       &   \cdots  &   0   \\
                    &           &   \ddots  &   \ddots  &   \vdots  \\
                    &           &           &   \lambda &   1   \\
                    &           &           &           &   \lambda
        \end{array}\right]
    \end{equation}
\label{def:Jordan-cell}
\end{definition}

\begin{definition}
   Suppose
    \begin{equation}
    \label{eq:Lambda}
    \Lambda
    = \left[ \begin{array}{cc}
        \alpha  &   -\beta  \\
        \beta   &   \alpha
    \end{array}\right] \in \mathbf{M}_{2}\   ,\   \hbox{with}\  \beta > 0 \   .
    \end{equation}
    A \textbf{Jordan $\Lambda$-block of the first kind} $ J _n ( \Lambda ) \in \mathbf{M}_{2n}$ is a block upper triangular matrix of the form
    \begin{equation}
        J _n ( \Lambda ) = \left[ \begin{array}{ccccc}
                \Lambda   &   I _2    &   0 _2    &   \cdots  &   0 _2   \\
                        &   \Lambda       &   I _2    &   \cdots  &   0 _2   \\
                    &           &   \ddots  &   \ddots  &   \vdots  \\
                    &           &           &   \Lambda       &   I _2   \\
                    &           &           &           &   \Lambda
        \end{array}\right]\  .
    \end{equation}
\label{def:Jordan-Lambda-block-first}
\end{definition}
In the remainder of the article, $ \Lambda $ if not specified, is supposed to have the form in (\ref{eq:Lambda}) .

\begin{definition}
    Let $ \lambda \in \mathbb{R}  $ and $ n_1, \ldots, n_m $ be positive integers such that $ n = n_1 + \ldots + n_m $. A \textbf{Jordan matrix} $ J _{ n _1, \ldots, n _m } ( \lambda ) \in \mathbf{M}_{n}$ is a block diagonal matrix
    \begin{equation}
        J _{ n _1, \ldots, n _m } ( \lambda ) = \diagonal{
            J_{n_1} ( \lambda ),
            \ldots,
            J_{n_m} ( \lambda )
        }
    \end{equation}
    where $J_{n_i} ( \lambda )$ are Jordan cells for all $i = 1, \ldots, m $.
\label{def:Jordan-matrix}
\end{definition}

\begin{definition}
    Let $ n _1, \ldots, n _m $ be positive integers such that $ n = n _1 + \ldots + n _m $.
    A \textbf{Jordan $\Lambda$-block of the second kind} $ J _{ n _1, \ldots, n _m } ( \Lambda ) \in \mathbf{M}_{2n}$ is a block diagonal matrix
    \begin{equation}
        J_ {n _1, \ldots, n _m } ( \Lambda ) = \diagonal{
            J _{n_1} ( \Lambda ),
            \ldots,
            J _{n _m} ( \Lambda )
        }
    \end{equation}
    where $J_{n _i} ( \Lambda ) $ are Jordan $\Lambda$-blocks of the first kind for all $ i = 1, \ldots, m $.
\label{def:Jordan-Lambda-block-second}
\end{definition}

\begin{definition}
    Let $ \lambda _i \in\mathbb{R} $, $i \in \{ 1,2,\cdots ,p \}$ and
    \begin{equation}
    \Lambda _k
    = \left[ \begin{array}{cc}
        \alpha _k   &   -\beta _k  \\
        \beta _k    &   \alpha _k
    \end{array}\right] \in \mathbf{M}_{2} ,\  \  k \in \{ 1,2,\cdots ,q \}
    \end{equation}
    with $ \beta _k > 0 $, such that all scalars and matrices are distinct.
    Let $ m ^i _1, \ldots, m ^i _{r _i} $ and $ n ^k _1, \ldots, n ^k _{s _k} $ be positive integers such that $ m ^i = m ^i _1 + \ldots + m ^i _{r _i} $, $ n ^k = n ^k _1 + \ldots + n ^k _{r _k} $, $ m = m ^1 + \ldots + m ^p $ and $ n = n ^1 + \ldots + n ^q $.
    A \textbf{generalized Jordan matrix} $ J \in \mathbf{M}_{m + 2n}$ is defined as a block diagonal matrix of the form
    \begin{equation}
    \fl J = \diagonal{
        J _{m ^1 _1, \ldots, m ^1 _{r _1}} ( \lambda _1 ),
        \ldots,
        J _{m ^p _1, \ldots, m ^p _{r _p}} ( \lambda _p ),
        J _{n ^1 _1, \ldots, n ^1 _{s _1}} ( \Lambda _1 ),
        \ldots,
        J _{n ^q _1, \ldots, n ^q _{s _q}} ( \Lambda _q )
    }
    \end{equation}
    where $ J _{m ^i _1, \ldots, m ^i _{r _i}} ( \lambda _i ) $ are Jordan matrices for all $ \lambda _i $, and $ J _{n ^k _1, \ldots, n ^k _{s _k}} ( \Lambda _k ) $ are Jordan $\Lambda$-blocks of the second kind for all $ i \in \{ 1, \ldots, p \} $ and $ k \in \{ 1, \ldots, q \} $.
\label{def:Jordan-generalized-matrix}
\end{definition}

\begin{theorem}[from \cite{horn_johnson_1985}]
    Each $ A \in \mathbf{M}_{n}$ is similar via a real similarity transformation matrix, to a generalized Jordan matrix of the form given in Definition \ref{def:Jordan-generalized-matrix} in which
    the scalars $ \lambda _1, \ldots, \lambda _p $ are real eigenvalues of $ A $, and its complex conjugate eigenvalues $ \alpha _k \pm i \beta _k $ are represented by the matrices $ \Lambda _k $ for all $ k \in \{ 1, \ldots, q \} $.
\end{theorem}


\begin{theorem}
\label{theorem: ia jordan block first kind real eigenvalue}
    Let $ \lambda \in \mathbb{R} $ and $J _n ( \lambda )$ be a Jordan cell. Then $\mathcal{I}(J_n (\lambda))$ coincides with the set of all real square matrices of order $n$ which are of the form
    \begin{equation}
        \left[ \begin{array}{cccc}
            x _1    &   x _2    &   \cdots  &   x _n    \\
            x _2    &   x _3    &   \cdots  &   0   \\
            \vdots  &   \vdots  &   \ddots  &   \vdots  \\
            x _n    &   0       &   \cdots  &   0
       \end{array}\right] \  .
    \end{equation}
\end{theorem}

\noindent \textit{Proof:}

    Let be
    \begin{equation}
        X = \left[ \begin{array}{ccc}
            x _{11} &   \cdots  &   x _{1n} \\
            \vdots  &   \ddots  &   \vdots  \\
            x _{1n} &   \cdots  &   x _{nn}
        \end{array}\right] \in \mathbf{M}_{n}
    \end{equation}
    The intertwining relation $J _n ( \lambda ) X = X J ^T _n ( \lambda )$ implies the following:
    \begin{equation}
    \label{eq:properties_antidiagonal}
    x _{i+1,j} = x _{i,j+1}\  \hbox{for}\  i, j \in \{ 1, \ldots, n - 1 \},
    \  \   x _{k n} = 0\  \hbox{for}\  k \in \{ 2, \ldots, n \}  .
    \end{equation}
    Equations (\ref{eq:properties_antidiagonal}) mean that all entries of any antidiagonal of $X$, are equal, and, that
    all antidiagonals below the main antidiagonal are zero.

\rightline{$\square$}

\begin{lemma}
\label{lemma: J m X = X J n lambda m < n}
    Let $m$ and $n$ be two positive integers such that $ m < n $ and let $ X \in \mathbf{M}_{m \times n}$. Then any Jordan cells $J_m, J_n$ satisfy that
    \begin{equation}
        J _m ( \lambda ) X = X J ^T _n ( \lambda ) \iff X = \left[ \begin{array}{cc} Y & 0 \end{array}\right], \  Y \in \mathcal{I} ( J _m ( \lambda ) ) \  .
    \end{equation}
\end{lemma}

\noindent \textit{Proof:}

    Let $p$ be a positive integer such that $ m \leq p \leq n -1 $.
    If $ J _m ( \lambda ) X = X J ^T _n ( \lambda ) $, then $ J _m ( 0 ) X = X J ^T _n ( 0 ) $, so that $ J ^p _m ( 0 ) X = X ( J ^p _n ( 0 ) ) ^T $.
    Therefore $ X ( J ^p _n ( 0 ) ) ^T = 0 $.
    If $ p = n - 1 $, then $ x _{i n} = 0 $ for each $ i \in \{ 1, \ldots, m \} $.
    Proceeding analogously in decreasing order to $ p = m $ we obtain
    \begin{equation}
        X = \left[ \begin{array}{cccccc}
            x_{11}  &   \cdots  &   x _{1m} &   0   &   \cdots  &   0
        \\  \vdots  &   \ddots  &   \vdots  &   \vdots  &   \ddots  &   \vdots
        \\  x_{m1}  &   \cdots  &   x _{mm} &   0   &   \cdots  &   0
        \end{array}\right]
    \end{equation}
    Hence, we can write $X = \left[ \begin{array}{cc} Y & 0 \end{array}\right] $, where $Y \in \mathbf{M}_{m}$.
    If we partition $ J _n ( \lambda ) $ as
    \begin{equation}
        J _n ( \lambda ) = \left[ \begin{array}{cc}
            J _m ( \lambda )    &   E _{m,1}    \\
            0   &   J _{n-m} ( \lambda )
        \end{array}\right]
    \end{equation}
    where
    \begin{equation}
        E _{m,1} = \left[ \begin{array}{cccc}
            0   &   0   &   \cdots  &   0   \\
            \vdots  &   \vdots  &   \ddots  &   \vdots  \\
            0   &   0   &   \cdots  &   0   \\
            1   &   0   &   \cdots  &   0   \\
        \end{array}\right] \in \mathbf{M}_{m \times (n - m)}
    \end{equation}
    then the intertwining relation $ J _m ( 0 ) X = X J ^T _n ( 0 ) $ implies $ J _m ( \lambda ) Y = Y J ^T _m ( \lambda ) $, so that $ Y \in \mathcal{I} ( J _m ( \lambda ) ) $.

    Now, let $X = \left[ \begin{array}{cc} Y & 0 \end{array}\right] \in \mathbf{M}_{m \times n}$ and let $ Y \in \mathbf{M}_{m}$.
    If $ Y \in \mathcal{I} ( J _m ( \lambda ) ) $, then $ J _m ( \lambda ) Y = Y J ^T _m ( \lambda ) $, which implies that $ J _m ( \lambda ) X = X J ^T _n ( \lambda ) $.

\rightline{$\square$}

\begin{lemma}
\label{lemma: J m X = X J n lambda m > n}
    Let $m$ and $n$ be two positive integers such that $ m > n $ and let $ X \in \mathbf{M}_{m \times n}$. Then any Jordan cells $J_m, J_n$ satisfy that
    \begin{equation}
        J _m ( \lambda ) X = X J ^T _n ( \lambda ) \iff X = \left[ \begin{array}{c} Y \\ 0 \end{array}\right], \  Y \in \mathcal{I} ( J _n ( \lambda ) )  \  .
    \end{equation}
\end{lemma}

\noindent \textit{Proof:}

    We rewrite $ J _m ( \lambda ) X = X J ^T _n ( \lambda ) $ as $ J _n ( \lambda ) X ^T = X ^T J ^T _m ( \lambda ) $.
    By Lemma \ref{lemma: J m X = X J n lambda m < n} we have $X^T = \left[ \begin{array}{cc} Y & 0 \end{array}\right] $ with $Y \in \mathcal{I} ( J _n ( \lambda ) ) $, hence $ X = \left[ \begin{array}{c} Y \\ 0 \end{array}\right] $.

\rightline{$\square$}

\begin{theorem}
\label{theorem: ia jordan block second kind real eigenvalue}
    Let $ n _1, \ldots, n _m $ be positive integers such that $ n = n _1 + \ldots + n _m $,
    $ \lambda \in \mathbb{R}  $, and let $ J _{ n _1, \ldots, n _m } ( \lambda ) \in \mathbf{M}_{n}$ be a Jordan matrix. Then every matrix $X \in \mathcal{I}( J _{ n _1, \ldots, n _m } ( \lambda ) )$ is a block matrix of the form
    \begin{equation}
        X = \left[ \begin{array}{ccc}
            X _{1 1}    &   \cdots  &   X _{1 m}    \\
            \vdots      &   \ddots  &   \vdots  \\
            X _{m 1}    &   \cdots  &   X _{m m}
        \end{array}\right]
    \end{equation}
    where for each $i, j \in\{ 1, \ldots, m\} $, $ X _{i j} \in \mathbf{M}_{n _i\times n _j}$ satisfies $X ^T _{i j}  = X _{j i}  $ and are of the following form:
    \begin{enumerate}
        \item   If $n_i = n_j$ then $X_{ij} \in\mathcal{I}( J_{n_i}(\lambda ))$.
        \item   If $n_i < n_j$ then $X_{ij} = \left[ \begin{array}{cc} Y _{ij} & 0 \end{array}\right] $ with $Y_{ij}\in\mathcal{I}(J_{n_i}(\lambda ))$.
        \item   If $n_i > n_j$ then $X_{ij} = \left[ \begin{array}{c} Y _{ij} \\ 0 \end{array}\right] $ with $Y_{ij} \in\mathcal{I}(J_{n_j}(\lambda ))$.
    \end{enumerate}
\end{theorem}

\noindent \textit{Proof:}
    Let $ i, j \in \{ 1, \ldots, m\} $.
    If $X \in \mathcal{I}( J _{ n _1, \ldots, n _m } ( \lambda ) )$ then $ J _{ n _1, \ldots, n _m } ( \lambda ) X = X J ^T _{ n _1, \ldots, n _m } ( \lambda ) $ and $ X ^T = X $, so that $ J _{n _i} (\lambda) X _{i j} = X _{i j} J _{n _j} ^T (\lambda) $ and $ X _{j i} = X ^T _{i j} $.
    If $n _i = n _j $ then $X _{i j}$ is symmetric, so that $X _{i j} \in \mathcal{I}( J _{n _i} (\lambda) )$.
    By Lemma \ref{lemma: J m X = X J n lambda m < n} we have that $X_{ij} = \left[ \begin{array}{cc} Y_{ij} & 0 \end{array}\right] $ with $Y_{ij} \in \mathcal{I}(J_{n_i}(\lambda ))$ for $n_i < n_j$.
    For $n _i > n _j$, by Lemma \ref{lemma: J m X = X J n lambda m > n} we find $ X_{ij} = \left[ \begin{array}{c} Y _{ij} \\ 0 \end{array}\right] $ with $Y_{ij}\in\mathcal{I}(J_{n_j}(\lambda ))$.

\rightline{$\square$}


\begin{lemma}
\label{lemma: ia Lambda}
For any $\Lambda  = \left[ \begin{array}{cc} \alpha & -\beta \\ \beta & \alpha \end{array} \right] \in \mathbf{M}_{2}$ with $\beta >0$, $\mathcal{I}(\Lambda)$ coincides with the set of all real symmetric 2x2- matrices of the form
\begin{equation}
   \left[ \begin{array}{cc} a & b \\ b & -a \end{array} \right]
\end{equation}
\end{lemma}

\noindent \textit{Proof:}
    For any
    \begin{equation}
    X = \left[ \begin{array}{cc} x _1 & x _2 \\ x _3 & x _4 \end{array}\right] \in\mathbf{M}_{2}\  ,
    \end{equation}
    from the intertwining relation $\Lambda X = X \Lambda ^T$ together with $ \beta > 0 $ we get $ x _3 = x _2 $ $ x _4 = - x _1 $.

\rightline{$\square$}

\begin{lemma}
\label{lemma: linear eq ia}
    Let $ X \in \mathbf{M}_{2}$,
   $\Lambda  = \left[ \begin{array}{cc} \alpha & -\beta \\ \beta & \alpha \end{array} \right] \in \mathbf{M}_{2}$ with $\beta > 0$.
    Any $ Y \in \mathcal{I} ( \Lambda ) $ satisfies that
    \begin{equation}
        \Lambda X = X \Lambda ^T + Y \iff X \in \mathcal{I} ( \Lambda ), Y = 0 \  .
    \end{equation}
\end{lemma}

\noindent \textit{Proof:}

    Let $X = \left[ \begin{array}{cc} x & y \\ z & t \end{array}\right] \in\mathbf{M}_{2}$.
    If $Y \in \mathcal{I} ( \Lambda ) $, then $ Y = \left[ \begin{array}{cc} a & b \\ b & -a \end{array}\right] \in\mathbf{M}_{2}$.
    The relation $ \Lambda X = X \Lambda ^T + Y $ together with $ \beta > 0 $ implies $ a = b = 0 $, $ z = y $ and $ t = -x $.
    On the other hand, if $ X \in \mathcal{I} ( \Lambda ) $ then $ \Lambda X = X \Lambda ^T $, hence $ Y = 0 $.

\rightline{$\square$}

\begin{theorem}
\label{theorem: ia jordan block first kind complex conjugate eigenvalues}
    Let $J _n ( \Lambda )$ be a Jordan $\Lambda$-block of the first kind with complex conjugate eigenvalues represented by $\Lambda  = \left[ \begin{array}{cc} \alpha & -\beta \\ \beta & \alpha \end{array} \right] $.
    Then
    \begin{equation}
    \fl \mathcal{I}( J _n ( \Lambda ) ) = \  \left\{ \left[ \begin{array}{cccc}
            X _1    &   X _2    &   \cdots  &   X _n    \\
            X _2    &   X _3    &   \cdots  &   0 _2   \\
            \vdots  &   \vdots  &   \ddots  &   \vdots  \\
            X _n    &   0 _2    &   \cdots  &   0 _2
        \end{array}\right] \in\mathbf{M}_{n} :\,  X _i \in\mathcal{I}(\Lambda )\  \hbox{for}\   i \in \{1, \ldots, n \} \right\}
    \end{equation}
\end{theorem}

\noindent \textit{Proof:}

    For $ n = 1 $ see Lemma \ref{lemma: ia Lambda}.
    For $ n = 2 $, let
    \begin{equation}
        \mathfrak{X} = \left[ \begin{array}{cc}
            X   &   Y
        \\  Z   &   T
        \end{array}\right] \in \mathbf{M}_{4}
    \end{equation}
    where $ X, Y, Z, T \in \mathbf{M}_{2}$.
    The intertwining relation $ J _2 ( \Lambda ) \mathfrak{X} = \mathfrak{X} J ^T _2 ( \Lambda ) $ implies
    \begin{equation}
    \label{eq:ia_Lamba_eq_1_4}
    \eqalign{
    \Lambda T     &=  T \Lambda^T  \\
    \Lambda Z     &=  Z \Lambda^T + T   \\
    \Lambda Y + T &=  Y \Lambda^T   \\
    \Lambda X + Z &=  X \Lambda^T + Y
    } \end{equation}  
    The first one of the Equations (\ref{eq:ia_Lamba_eq_1_4}) implies that $ T \in \mathcal{I} ( \Lambda ) $.
    Applying Lemma \ref{lemma: linear eq ia} to the second and third ones of Equations (\ref{eq:ia_Lamba_eq_1_4}), we obtain $Y, Z \in \mathcal{I} ( \Lambda ) $ and $ T = 0 $.
    Using Lemma \ref{lemma: linear eq ia} in the last one of Equations (\ref{eq:ia_Lamba_eq_1_4}) we find that $X \in \mathcal{I} ( \Lambda ) $ and $ Z = Y $. Thus
    \begin{equation}
        \mathcal{I}( J _2 ( \Lambda ) ) = \  \left\{
        \left[ \begin{array}{cc} X & Y \\ Y & 0 \end{array}\right] :
        X, Y \in \mathcal{I} ( \Lambda ) \right\}  \  .
    \end{equation}
    Now, assume that the property is true for $n$ and let us prove that it is satisfied for $n + 1$.

    $ J _{n + 1} ( \Lambda ) $ can be partitioned as follows:
    \begin{equation}
        J _{n + 1} ( \Lambda ) = \left[ \begin{array}{cc}
            \Lambda &   E_1
        \\  0       &   J _n ( \Lambda ) \end{array}\right]
    \end{equation}
    where $E _1 = \left[ \begin{array}{cccc} I_2 & 0 & \cdots & 0 \end{array}\right] \in \mathbf{M}_{2\times 2n}$.
    Let
    \begin{equation}
        \mathfrak{X} = \left[ \begin{array}{cc}
            T   &   Y
        \\  Z   &   X
        \end{array}\right] \in \mathbf{M}_{2(n+1)}
    \end{equation}
    where $ X \in \mathbf{M}_{2n}$, $Y = \left[ \begin{array}{ccc} Y_1 & \cdots & Y_n \end{array}\right],
    Z ^T = \left[ \begin{array}{ccc} Z_1^T & \cdots & Z_n^T \end{array}\right]$, and all matrices $T, Y_1, \ldots, Y_n, Z_1, \ldots, Z_n$ belong to $\mathbf{M}_{2}$.

    If $ J _{n + 1} ( \Lambda ) \mathfrak{X} = \mathfrak{X} J ^T _{n + 1} ( \Lambda ) $ then
    \begin{equation}
    \label{eq:ia_Lambda_eq_all_n+1}
    \eqalign{
    J _n ( \Lambda ) X &= J ^T _n ( \Lambda ) \\
    J _n ( \Lambda ) Z &= \Lambda Y + E _1 X  \\
    Z \Lambda ^T + X E _1 ^T &= Y J ^T _n ( \Lambda )  \\
    \Lambda T + E _1 Z &= T \Lambda ^T + Y E _1 ^T
    } \end{equation}
    From the one of the Equations (\ref{eq:ia_Lambda_eq_all_n+1}) we obtain $ X \in \mathcal{I} ( J _n ( \Lambda ) ) $.
    By the induction hypothesis we can write
    \begin{equation}
        X = \left[ \begin{array}{cccc}
            X_1 &   X_2 &   \cdots  &   X_n
        \\  X_2 &   X_3 &   \cdots  &   0
        \\  \vdots  &   \vdots  &   \ddots  &   \vdots
        \\  X_n &   0   &   \cdots  &   0
        \end{array}\right]
    \end{equation}
    The second and third one of Equations (\ref{eq:ia_Lambda_eq_all_n+1}) imply that
    \begin{equation}  
    \label{eq:ia Lamba n+1}
    \eqalign{
    \Lambda Z_n &= Z_n \Lambda ^T + X_n\, ,
    \  \Lambda Y_n + X_n = Y_n \Lambda ^T \\
    \Lambda Z_{n-1} + Z_n &= Z_{n-1} \Lambda ^T + X_{n-1}\, ,
    \  \Lambda Y_{n-1} + X_{n-1} = Y_{n-1} \Lambda ^T + Y_n \\
     &\vdots   \\
    \Lambda Z_2 + Z_3  &= Z_2 \Lambda ^T + X_2\, ,
    \  \Lambda Y_2 + X_2 = Y_2 \Lambda ^T + Y_3 \\
    \Lambda Z_1 + Z_2 &= Z_1 \Lambda ^T + X_1\, ,
    \  \Lambda Y_1 + X_1 = Y_1 \Lambda ^T + Y_2\  \  .
    } \end{equation}
    Using Lemma \ref{lemma: linear eq ia} for Equations (\ref{eq:ia Lamba n+1}), one obtains $ X_n = 0, Z_n = Y_n = X_{n-1}, \ldots, Z_3 = Y_3 = X_2, Z_2 = Y_2 = X_1 $ and $ Z_1, Y_1 \in \mathcal{I} ( \Lambda ) $.
    Then, applying Lemma \ref{lemma: linear eq ia} in the last one of Equations (\ref{eq:ia_Lambda_eq_all_n+1}) gives $ Z _1 = Y _1 $.
    Therefore,
    \begin{equation}
        \mathfrak{X} = \left[ \begin{array}{ccccc}
            T       &   Y_1 &   X_1 &   \cdots  &   X_{n-1}
        \\  Y_1     &   X_1 &   X_2 &   \cdots  &   0
        \\  X_1     &   X_2 &   X_3 &   \cdots  &   0
        \\  \vdots  &   \vdots  &   \vdots  &   \ddots  &   \vdots
        \\  X_{n-1} &   0   &   0   &   \cdots  &   0
        \end{array}\right]  \  .
    \end{equation}
    Finally, it is obvious that $\mathfrak{X} ^T = \mathfrak{X} $.

\rightline{$\square$}

\begin{lemma}
\label{lemma: J m X = X J T n Lambda m < n}
    Let $m$ and $n$ be two positive integers such that $ m < n $, and let $ X \in \mathbf{M}_{2m\times 2n}$. Then
    \begin{equation}
        J _m ( \Lambda ) X = X J ^T _n ( \Lambda ) \  \hbox{if and only if}\  X = \left[ \begin{array}{cc} Y & 0 \end{array}\right] , Y \in \mathcal{I} ( J _m ( \Lambda ) )
    \end{equation}
\end{lemma}

\noindent \textit{Proof:}

    Let
    \begin{equation}
        \mathfrak{X} = \left[ \begin{array}{ccc}
            X _{1 1}    &   \cdots  &   X _{1 n}    \\
            \vdots      &   \ddots  &   \vdots  \\
            X _{m 1}    &   \cdots  &   X _{m n}
        \end{array}\right] \in \mathbf{M}_{2m \times 2n}
    \end{equation}
    where $ X _{i j} \in \mathbf{M}_{2}$ for all $ i \in \{ 1, \ldots, m \} $ and $ j \in \{ 1, \ldots, n \} $.
    The intertwining relation $ J _m ( \Lambda ) \mathfrak{X} = \mathfrak{X} J ^T _n ( \Lambda ) $ implies the equations
    \begin{equation}
    \label{eq:Lambda_and_X mn}
    \eqalign{
    \Lambda X _{m n}                 &= X _{m n} \Lambda ^T                   \\
    \Lambda X _{i n} + X _{i + 1, n} &= X _{i n} \Lambda ^T                   \\
    \Lambda X _{m j}                 &= X _{m j} \Lambda ^T + X _{m, j + 1 }  \\
    \Lambda X _{i j} + X _{i + 1, j} &= X _{i j} \Lambda ^T + X _{i, j + 1 }  \\
    } \end{equation}
    for each $ i \in \{ 1, \ldots, m - 1 \}$ and $ j \in \{ 1, \ldots, n - 1 \} $.
    The first one of the Equations (\ref{eq:Lambda_and_X mn}) implies $ X _{m n} \in \mathcal{I} ( \Lambda ) $.
    Moreover, applying Lemma \ref{lemma: linear eq ia} to the second and third ones of Equations (\ref{eq:Lambda_and_X mn}), one obtains that $X _{1 n} \in \mathcal{I} ( \Lambda ), X _{2 n} = \cdots = X _{m n} = 0 $ and $ X _{m 1} \in \mathcal{I} ( \Lambda ), X _{m 2} = \cdots = X _{m n} = 0 $, respectively.
    Now, we will only use the last one of Equations (\ref{eq:Lambda_and_X mn}).
    Let us write $ X _n $ instead of $ X _{1 n} $. For $ j = n - 1 $, taking into account Lemma \ref{lemma: linear eq ia}, we get $ X _{1, n - 1} \in \mathcal{I} ( \Lambda ), X _{2, n - 1} = X _n, X _{3, n - 1} = \cdots X _{m, n - 1} = 0 $.
    Also, let us write $ X _{n - 1 } $ instead of $ X _{1, n - 1} $.
    Proceeding analogously as before, it turns out that
    \begin{equation}
        \mathfrak{X} = \left[ \begin{array}{cccccc}
            X _1    &   \cdots  &   X _{n - m + 1} &   X _{n - m + 2} &   \cdots  &   X _n      \\
            \vdots      &   \ddots  &   \vdots  &   \vdots  &   \ddots    &   \vdots  \\
            X _m    &   \cdots  &   X _n   &   0    &   \cdots  &   0
        \end{array}\right]
    \end{equation}
    where $ X _{i + j - 1} = X _{i j} $ with $ i + j \leq n + 1 $.
    However, we had found that only $ X _m $ is non-zero in the last row.
    Then, $ X _{m + 1} = \cdots = X _n = 0 $, so that
    \begin{equation}
        \mathfrak{X} = \left[ \begin{array}{cccccc}
            X _1    &   \cdots  &   X _m    &   0       &   \cdots  &   0   \\
            \vdots  &   \ddots &   \vdots  &   \vdots  &   \ddots &   \vdots   \\
            X _m    &   \cdots  &   0       &   0       &   \cdots  &   0
        \end{array}\right]  \  .
    \end{equation}

    On the other hand, we may partition $ J _n ( \Lambda ) $ as
    \begin{equation}
        J _n ( \Lambda ) = \left[ \begin{array}{cc}
            J _m ( \Lambda )  &   E _{m,1}    \\
            0                       &   J _{n - m} ( \Lambda )
        \end{array}\right]
    \end{equation}
    where
    \begin{equation}
        E _{m,1} = \left[ \begin{array}{cccc}
            0   &   0   &   \cdots  &   0   \\
            \vdots  &   \vdots  &   \ddots  &   \vdots  \\
            0   &   0   &   \cdots  &   0   \\
            I_2 &   0   &   \cdots  &   0   \\
        \end{array}\right] \in \mathbf{M}_{2m \times 2(n-m)}  .
    \end{equation}
    Let $ X \in \mathbf{M}_{2m}$ and $ \mathfrak{X} = \left[ \begin{array}{cc} X & 0 \end{array}\right] \in \mathbf{M}_{2m \times 2n}$.
    If $ X \in \mathcal{I} ( J _m ( \Lambda ) ) $ then $ J _m ( \Lambda ) X = X J ^T _m ( \Lambda ) $, hence
    $ J _m ( \Lambda ) \mathfrak{X} = \mathfrak{X} J ^T _n ( \Lambda ) $.

\rightline{$\square$}

\begin{lemma}
\label{lemma: J m X = X J T n Lambda m > n}
    Let $m$ and $n$ be two positive integer numbers such that $ m > n $, and let $ X \in \mathbf{M}_{2m \times 2n}$. Then
    \begin{equation}
        J _m ( \Lambda ) X = X J ^T _n ( \Lambda ) \iff X = \left[ \begin{array}{c} Y \\ 0 \end{array}\right], \  Y\in \mathcal{I} ( J _n ( \Lambda ) )  \  .
    \end{equation}
\end{lemma}

\noindent \textit{Proof:}

   This can be proved in a similar way to the proof of Lemma \ref{lemma: J m X = X J T n Lambda m < n}.

\rightline{$\square$}

\begin{theorem}
\label{theorem: ia jordan block second kind complex conjugate eigenvalues}
    Let $ n _1, \ldots, n _m $ be positive integers such that $ n = n _1 + \ldots + n _m $, and let $ J _{ n _1, \ldots, n _m } ( \Lambda ) \in \mathbf{M}_{2n}$ be a Jordan matrix with complex conjugate eigenvalues represented by $\Lambda  = \left[ \begin{array}{cc} \alpha & -\beta \\ \beta & \alpha \end{array} \right] $.
    Then every matrix $X \in \mathcal{I}( J _{ n _1, \ldots, n _m } ( \Lambda ) )$ is a block matrix of the form
    \begin{equation}
        X = \left[ \begin{array}{ccc}
            X _{1 1}    &   \cdots  &   X _{1 m}    \\
            \vdots      &   \ddots  &   \ddots  \\
            X _{m 1} &   \cdots  &   X _{m m}
        \end{array}\right]
    \end{equation}
    where for each $i, j \in\{ 1, \ldots, m\} $, $X _{i j} \in \mathbf{M}_{2n_i \times 2n_j}$ and $X _{j i} = X ^T _{i j} $ which are of the following form:
    \begin{enumerate}
        \item   If $n_i = n_j$, then $X_{ij} \in \mathcal{I}( J _{n _i} ( \Lambda ) )$.
        \item   If $n_i < n_j$, then $X_{ij} = \left[ \begin{array}{cc} Y_{ij} & 0 \end{array}\right] $ with $Y_{ij} \in\mathcal{I}(J_{n_i} (\Lambda ) )$.
        \item   If $n_i > n_j$, then $X_{ij} = \left[ \begin{array}{c} Y_{ij} \\ 0 \end{array}\right] $ with $Y_{ij} \in\mathcal{I}(J_{n_j} (\Lambda ) )$.
    \end{enumerate}
\end{theorem}

\noindent \textit{Proof:}
    Let $ i, j \in\{ 1, \ldots, m \}$.
    If $X \in \mathcal{I}( J _{ n _1, \ldots, n _m } ( \Lambda ) )$ then $ J _{ n _1, \ldots, n _m } ( \Lambda ) X = X J ^T _{ n _1, \ldots, n _m } ( \Lambda ) $ and $ X ^T = X $, hence $ J _{n _i} ( \Lambda ) X _{ij} = X _{ij} J ^T _{n _j} ( \Lambda ) $ and $ X _{ji} = X ^T _{ij} $.
    It is obvious that $ X _{ij} \in \mathcal{I}( J _{n _i} ( \Lambda ) )$ for $n _i = n _j$.
    If $n _i < n _j$, by Lemma \ref{lemma: J m X = X J T n Lambda m < n} we have $X_{ij} = \left[ \begin{array}{cc} Y _{ij} & 0 \end{array}\right] $ with $Y_{ij} \in \mathcal{I}( J_{n_i} ( \Lambda ) )$.
    Using Lemma \ref{lemma: J m X = X J T n Lambda m > n} we find that $ X _{ij} = \left[ \begin{array}{c} Y _{ij} \\ 0 \end{array}\right] $ with $Y_{ij} \in \mathcal{I}( J_{n_j} ( \Lambda ) )$ for $n_i > n_j$.

\rightline{$\square$}


\begin{theorem}
\label{theorem: ia jordan matrix}
    Let $ J $ be a generalized Jordan matrix due to Definition \ref{def:Jordan-generalized-matrix}. Then $\mathcal{I}( J )$ is the set of all matrices $ \diagonal{ X_1, \ldots, X_p, Y _1, \ldots, Y_q } $
    such that $X_i \in\mathcal{I}(J_{m ^i _1, \ldots, m ^i _{r _i}}(\lambda_i )) $ for all  $i \in \{ 1, \ldots, p \}$ , and $Y_j \in\mathcal{I}(J_{n ^j _1, \ldots, n ^j _{s _j}}(\Lambda_j )) $ for all $j \in \{ 1, \ldots, q \}$.
\end{theorem}

\noindent \textit{Proof:}

    Let $i,j \in \{ 1, \ldots, p \} $ and $k,l \in \{ 1, \ldots, q \}$.
    Let
    \begin{equation}
        \mathfrak{X} = \left[ \begin{array}{cccccc}
            X _{1 1}    &   \cdots  &   X _{1 p}    &   Z _{1 1}    &   \cdots  &   Z _{1 q}    \\
            \vdots      &   \ddots  &   \vdots      &   \vdots      &   \ddots  &   \vdots  \\
            X _{p 1}    &   \cdots  &   X _{p p}    &   Z _{p 1}    &   \cdots  &   Z _{p q}    \\
            T _{1 1}    &   \cdots  &   T _{1 p}    &   Y _{1 1}    &   \cdots  &   Y _{1 q}    \\
            \vdots      &   \ddots  &   \vdots      &   \vdots      &   \ddots  &   \vdots  \\
            T _{q 1}    &   \cdots  &   T _{q p}    &   Y _{q 1}    &   \cdots  &   Y _{q q}
        \end{array}\right] \in \mathbf{M}_{m+2n}
    \end{equation}
    where $X _{i j} \in \mathbf{M}_{m ^i \times m ^j}$, $Y _{k l} \in \mathbf{M}_{2n^k \times 2n^l}$,
    $Z _{i l} \in \mathbf{M}_{m^i \times 2n^l}$ and $ T _{k j} \in \mathbf{M}_{2n^k \times m^j}$.
    If $J \mathfrak{X} =  \mathfrak{X} J ^T$, then
    \begin{eqnarray}
    \label{eq:properties_lambda}
    J _{m ^i _1, \ldots, m ^i _{r _i}} ( \lambda _i ) X _{i j}
    &=& X _{i j} J ^T _{m ^j _1, \ldots, m ^j _{r _j}} ( \lambda _j ) \\
    J _{n ^k _1, \ldots, n ^k _{s _k}} ( \Lambda _k ) Y _{k l}
    &=& Y _{k l} J ^T _{n ^l _1, \ldots, n ^l _{s _l}} ( \Lambda _l ) \\
    J _{m ^i _1, \ldots, m ^i _{r _i}} ( \lambda _i ) Z _{i k}
    &=& Z _{i k} J ^T _{n ^k _1, \ldots, n ^k _{s _k}} ( \Lambda _k )
    \end{eqnarray}
    Since the Jordan matrices do not have common eigenvalues, by the Sylvester's theorem on linear matrix equations \cite{Gantmacher59,horn_johnson_1985} we have $ X_{i j} = 0 $ for $i \neq j$, $Y _{k l} = 0$ for $k \neq l$, $Z _{i l} = 0$ and $ T _{k j} = 0 $.
    Furthermore, if $ \mathfrak{X} $ is symmetric, $ X _{ii} $ and $ Y _{k k} $ are also symmetric, so that $X _{i i} \in \mathcal{I}( J _{m ^i _1, \ldots, m ^i _{r _i}} ( \lambda _i ) )$ and $Y _{k k} \in \mathcal{I}( J _{n ^k _1, \ldots, n ^k _{s _k}} ( \Lambda _k ) )$.

 \rightline{$\square$}


\begin{theorem}
\label{theorem: det ia real eigenvalue}
    For any $x _i \in \mathbb{R}$, $i \in \{ 1, \ldots, n \}$, the following determinant formula holds:
    \begin{equation}
        \left| \begin{array}{cccc}
            x _1    &   x _2    &   \cdots  &   x _n    \\
            x _2    &   x _3    &   \cdots  &   0       \\
            \vdots  &   \vdots  &   \ddots  &   \vdots  \\
            x _n    &   0       &   \cdots  &   0       \end{array}\right|
        = (-) ^{n(n-1)/2} \cdot x^n_n
    \end{equation}
\end{theorem}

\noindent \textit{Proof:}

    Let
    \begin{equation}
         K _n = \left[ \begin{array}{ccc} & & 1 \\ & \cdots & \\ 1 & &  \end{array}\right] \in \mathbf{M}_{n}
    \end{equation}
    be the exchange matrix.
    Using that $ \det K _n = (-) ^{n (n-1)/2} $ we find
    \begin{equation}
    \fl \left| \begin{array}{cccc}
            x _1    &   x _2    &   \cdots  &   x _n    \\
            x _2    &   x _3    &   \cdots  &   0       \\
            \vdots  &   \vdots  &   \ddots  &   \vdots  \\
            x _n    &   0       &   \cdots  &   0       \end{array}\right| 
        =  \left| \begin{array}{cccc}
            x _n    &   x _{n-1}    &   \cdots  &   x _1    \\
            0       &   x _n        &   \cdots  &   x _2    \\
            \vdots  &   \vdots      &   \ddots  &   \vdots  \\
            0       &   0           &   \cdots  &   x _n    \end{array}\right| \cdot 
         \left| \begin{array}{cccc}
              &        &   & 1  \\
              &        & 1 &    \\
              & \cdots &   &    \\
             1 &       &   &    \end{array}\right| 
         = (-) ^{n(n-1)/2} \cdot x^n_n
    \end{equation}

\rightline{$\square$}

\begin{theorem}
\label{theorem: det ia complex conjugate eigenvalues}
    Let $ X _i \in \mathcal{I}(Z)$ for $i \in \{ 1, \ldots, n \}$. The following determinant formula holds:
    \begin{equation}
        \left| \begin{array}{cccc}
            X _1    &   X _2    &   \cdots  &   X _n    \\
            X _2    &   X _3    &   \cdots  &   0       \\
            \vdots  &   \vdots  &   \ddots  &   \vdots  \\
            X _n    &   0       &   \cdots  &   0       \end{array}\right|
        = |X _n| ^n
    \end{equation}
\end{theorem}

\noindent \textit{Proof:}  

    By means of the properties of the determinants we have
	\begin{equation}
         \left| \begin{array}{ccc}
                &        & I_2 \\
                & \cdots &     \\
            I_2 &        &     \end{array}\right| 
            = ( -1 )^n  \left| \begin{array}{ccc}
                &        & K_2 \\
                & \cdots &     \\
            K_2 &        &     \end{array}\right| 
            = ( -1 )^n \det K_{2 n} = 1
        \end{equation}
    Then,
    \begin{equation}
    \fl \left| \begin{array}{cccc}
            X _1    &   X _2    &   \cdots  &   X _n    \\
            X _2    &   X _3    &   \cdots  &   0       \\
            \vdots  &   \vdots  &   \ddots  &   \vdots  \\
            X _n    &   0       &   \cdots  &   0       \end{array}\right|
         = \left| \begin{array}{cccc}
            X _n    &   X _{n-1}    &   \cdots  &   X _1    \\
            0       &   X _n        &   \cdots  &   X _2    \\
            \vdots  &   \vdots      &   \ddots  &   \vdots  \\
            0       &   0           &   \cdots  &   X _n    \end{array}\right|
        \cdot \left| \begin{array}{cccc}
                &        &     & I_2  \\
                &        & I_2 &      \\
                & \cdots &     &      \\
            I_2 &        &     &      \end{array}\right|
         = \  | X _n | ^n
    \end{equation}

\rightline{$\square$}

\section{Computing one-dimensional subspaces}   %
\label{sec:computing_1D-subspaces}

Now we will apply the properties of Jordan matrices deduced in the last section, to obtain knowledge about the matrix $g$.

\begin{theorem}
\label{theorem: e J jordan block first kind real eigenvalue}
    Let $ \lambda \in \mathbb{R}  $, and let $ J _n ( \lambda ) $ be a Jordan cell.
    Suppose $g \in \mathbf{Sym}_n$ as a matrix function such that $g_{, \xi} = J _n ( \lambda ) g$.
    Then
    \begin{equation}
    \label{eq:metric form lambda block 1}
        g _n ( \lambda ) = \left[ \begin{array}{cccc}
            X _1    &   X _2    &   \cdots  &   X _n    \\
            X _2    &   X _3    &   \cdots  &   0    \\
            \vdots  &   \vdots  &   \ddots  &   \vdots    \\
            X _n    &   0       &   \cdots  &   0    \\
        \end{array}\right]
    \end{equation}
    where
    \begin{equation}
    \label{eq:sol diff eq lambda}
        X _i ( \xi ) = e ^{\lambda \xi} \sum _{j = 0} ^{n -i} \frac{\xi ^j}{j !} C _{i + j}
    \end{equation}
    and $ C _i $ is constant for each $ i = 1, \ldots, n $.
\end{theorem}

\noindent \textit{Proof:}

    Applying $ g = g^T $ to $g_{, \xi} = J _n ( \lambda ) g$ we get $ J _n ( \lambda ) g= g J ^T _n ( \lambda ) $, then $g\in \mathcal{I} ( J _n ( \lambda ) ) $.
    By Theorem \ref{theorem: ia jordan block first kind real eigenvalue}, $g$ has the form given in Equation (\ref{eq:metric form lambda block 1}).
    From $g_{, \xi } = J _n ( \lambda ) g$ we obtain
    \begin{equation}
    \eqalign{
    X_{n, \xi} &= \lambda X_n  \\  
    X_{n-1, \xi} &= \lambda X_{n-1} + X_n \\
    & \vdots \\
    X_{1, \xi} &= \lambda X_1 + X_2
    } \end{equation}
    Integrating successively we get Equation (\ref{eq:sol diff eq lambda}).

\rightline{$\square$}

\begin{theorem}
\label{theorem: e J jordan block second kind real eigenvalue}
    Let $n_1, \ldots, n_m $ be positive integers such that $ n = n_1 + \ldots + n_m $.
    Let $ \lambda \in \mathbb{R}  $, and let $ J _{ n _1, \ldots, n _m } ( \lambda ) \in \mathbf{M}_{n}$ be a Jordan matrix.
    If $g\in \mathbf{Sym}_n$ is a matrix function such that $g_{, \xi } = J _{ n _1, \ldots, n _m } ( \lambda ) g$, then
    \begin{equation}
        g_{n_1 , \ldots , n_m} ( \lambda ) = \left[ \begin{array}{ccc}
            X_{1 1}    &   \cdots  &   X_{1 m}    \\
            \vdots      &   \ddots  &   \ddots  \\
            X_{m 1} &   \cdots  &   X_{m m}  \end{array}\right]
    \end{equation}
    where for each $i, j \in\{ 1, \ldots, m\} $, the matrix $X_{i j}$ satisfies $X_{i j}^T = X_{j i}$ and is defined as follows:
    \begin{enumerate}
        \item if $n_i = n_j$ then $X_{i j} = g_{n_i} ( \lambda ) $,
        \item if $n_i \leq n_j$ then $X_{i j} = \left[ \begin{array}{cc} g_{n_i} ( \lambda) &  0 \end{array}\right] $,
        \item if $n_i \geq n_j$ then $X_{i j} = \left[ \begin{array}{c} g_{n_j} ( \lambda) \\  0 \end{array}\right] $,
    \end{enumerate}
   where $g_{n_i} ( \lambda ) $ is defined as in Theorem \ref{theorem: e J jordan block first kind real eigenvalue}.
\end{theorem}

\noindent \textit{Proof:}

    Let $ i,j \in \{ 1, \ldots, m \} $.
    Applying $g = g^T $ to $g_{, \xi } = J_{n_1, \ldots, n_m} ( \lambda ) g$ we get $J_{n_1, \ldots, n_m} (\lambda ) g = g J^T _{n_1, \ldots, n_m } (\lambda )) $, then $g\in \mathcal{I} (J_{ n_1, \ldots, n_m } (\lambda ) ) $.
    By Theorem \ref{theorem: ia jordan block second kind real eigenvalue},
    \begin{equation}
        g ( \xi ) = \left[ \begin{array}{ccc}
            X _{1 1} ( \xi )    &   \cdots  &   X _{1 m} ( \xi )    \\
            \vdots              &   \ddots  &   \ddots  \\
            X _{m 1} ( \xi )    &   \cdots  &   X _{m m} ( \xi )
        \end{array}\right]
    \end{equation}
    where $X_{i j} \in \mathbf{M}_{n_i \times n_j}$ satisfies $ X_{j i} = X_{i j}^T $ and is of the following form:
    \begin{enumerate}
        \item   If $n _i = n _j$ then $X_{ij} \in \mathcal{I}( J_{n_i} (\lambda) )$.
        \item   If $n _i < n _j$ then $X_{ij} = \left[ \begin{array}{cc} Y_{n_i} & 0 \end{array}\right]$ with $Y_{n_i} \in \mathcal{I}( J_{n_i} (\lambda) )$.
        \item   If $n _i > n _j$ then $X_{ij} = \left[ \begin{array}{c} Y_{n_j} \\ 0 \end{array}\right]$ with $Y_{n_j} \in \mathcal{I}( J_{n_j} (\lambda) )$.
    \end{enumerate}
    From $g_{, \xi } = J_{ n_1, \ldots, n_m } (\lambda ) g$ we have $X_{i j , \xi } = J_{n_i} (\lambda) X_{i j}$.
    Observe that $ X_{j i, \xi } = \  (  J_{n_i} X_{i j} )^T = X_{j i} J^T_{n_i} (\lambda ) = J_{n_j} (\lambda ) X_{j i} $.
    By Theorem \ref{theorem: e J jordan block first kind real eigenvalue} we obtain
    \begin{enumerate}
        \item If $n_i = n_j$, then $X_{i j} = g_{n_i} ( \lambda )$.
        \item If $n_i \leq n_j$, then $X_{i j , \xi } = J_{n_i} (\lambda) X_{i j} $ implies $Y_{n_i , \xi } = J_{n_i} (\lambda) Y_{n_i }$, so that $Y_{n_i } = g_{n_i} ( \lambda )$.
        Therefore $X_{i j} = \left[ \begin{array}{cc} g_{n_i} (\lambda) & 0 \end{array}\right] $.
        \item If $n_i \geq n_j$, then $X_{i j , \xi } = J_{n_i} (\lambda ) X_{i j} = X_{i j} J_{n_j}^T (\lambda) $ implies $Y_{n_j , \xi } =
        Y_{n_j } J_{n_j}^T (\lambda ) = J_{n_j} (\lambda ) Y_{n_j } $, so that $Y_{n_j } = g_{n_j} ( \lambda )$.
        Hence $X_{i j} = \left[ \begin{array}{c} g_{n_j}( \lambda ) \\ 0 \end{array}\right] $.
    \end{enumerate}

\rightline{$\square$}


\begin{theorem}
\label{theorem: e J jordan block first kind complex conjugate eigenvalues}
    For any Jordan $\Lambda$-block of the first kind $J_n (\Lambda )$, if $g\in \mathbf{Sym}_{2n}$ is a matrix function such that $g_{, \xi } = J_n (\Lambda ) g$, then
       \begin{equation}
        g_n ( \Lambda ) = \left[ \begin{array}{cccc}
            Z _1    &   Z _2    &   \cdots  &   Z _n    \\
            Z _2    &   Z _3    &   \cdots  &   0    \\
            \vdots  &   \vdots  &   \ddots  &   \vdots    \\
            Z _n    &   0       &   \cdots  &   0    \\
        \end{array}\right]
    \end{equation}
    where
    \begin{equation}        
    \label{eq:sol diff eq Lambda}
    \eqalign{
    Z_l &= \left[ \begin{array}{cc}
            X _l   &   Y _l   \\
            Y _l   &   -X _l  \end{array}\right]  \\
    X_l (\xi ) &= e ^{\alpha \xi } \cos \beta \xi \sum _{k = 0} ^{n - l} \frac{\xi ^k}{k !} C _{k + l}
                  - e ^{\alpha \xi } \sin \beta \xi \sum _{k = 0} ^{n - l} \frac{\xi ^k}{k !} D _{k + l}  \\
    Y_l (\xi ) &= e ^{\alpha \xi } \cos \beta \xi \sum _{k = 0} ^{n - l} \frac{\xi ^k}{k !} D _{k + l}
                  + e ^{\alpha \xi } \sin \beta \xi \sum _{k = 0} ^{n - l} \frac{\xi ^k}{k !} C _{k + l}
    } \end{equation}
   and $C_l$, $D_l$ are constant for $l = 1, \ldots, n$.
\end{theorem}

\noindent \textit{Proof:}

    Let $ i,j \in \{ 1, \ldots, n \} $.
    Applying $g = g^T$ to $g_{, _\xi } = J_n (\Lambda )g$ we get $J_n (\Lambda ) g = g J^T_n (\Lambda )$, then $g \in \mathcal{I} (J_n (\Lambda ))$.
    By Theorem \ref{theorem: ia jordan block first kind complex conjugate eigenvalues} we can express
    \begin{equation}
        g (\xi ) = \left[ \begin{array}{cccc}
            Z _1 ( \xi )    &   Z _2 ( \xi )    &   \cdots  &   Z _n ( \xi )    \\
            Z _2 ( \xi )    &   Z _3 ( \xi )    &   \cdots  &   0    \\
            \vdots          &   \vdots          &   \ddots  &   \vdots    \\
            Z _n ( \xi )    &   0               &   \cdots  &   0    \\
        \end{array}\right]
    \end{equation}
    where
    \begin{equation}
	Z _i ( \xi ) = \left[ \begin{array}{cc}
		X _i (\xi )	&	Y _i (\xi )	 \\
		Y _i (\xi )	&	-X _i (\xi )   \end{array}\right]  \  .
    \end{equation}
    From $g_{, \xi } = J_n ( \Lambda ) g$ we have
    \begin{equation}
    \eqalign{        
    Z_{n, \xi} &= \lambda Z_n \\
    Z_{n-1, \xi} &= \lambda Z_{n-1} + Z_n \\
    & \vdots \\
    Z_{1, \xi} &= \lambda Z_1 + Z_2
    } \end{equation}
    Integrating successively we get
    \begin{equation}
    Z_i (\xi ) = e ^{\xi \Lambda} \sum _{j = 0} ^{n -i} \frac{\xi ^j}{j !} \mathfrak{C} _{i + j} \  ,
    \  \  \mathfrak{C} _i = \left[ \begin{array}{cc} C_i & D_i \\ D_i & -C_i \end{array}\right] \  .
    \end{equation}
    Using
    \begin{equation}
        e ^{\xi \Lambda}
    =   e ^{\alpha \xi} \left[ \begin{array}{cc}
            \cos \beta \xi  &   -\sin \beta \xi \\
            \sin \beta \xi  &   \cos \beta \xi  \end{array}\right]
    \end{equation}
    we obtain equation (\ref{eq:sol diff eq Lambda}).

\rightline{$\square$}

\begin{theorem}
\label{theorem: e J jordan block second kind complex conjugate eigenvalues}
    Let $n_1, \ldots, n_m $ be positive integers such that $ n = n_1 + \ldots + n_m $, and let $ J_{ n_1, \ldots, n_m } ( \Lambda ) \in \mathbf{M}_{2 n}$ be a Jordan $\Lambda$-block of the second kind.
    If $g\in \mathbf{Sym}_{2n}$ is a matrix function such that $g_{, \xi } = J_{ n_1, \ldots, n_m } (\Lambda ) g$, then
    \begin{equation}
        g_{n_1 , \ldots , n_m} ( \Lambda ) = \left[ \begin{array}{ccc}
            X _{1 1}    &   \cdots  &   X _{1 m}    \\
            \vdots      &   \ddots  &   \ddots  \\
            X _{m 1} &   \cdots  &   X _{m m}  \end{array}\right]  \  .
    \end{equation}
    The matrices $X_{i j}$ satisfy $X_{i j}^T = X_{j i}$ and are defined as follows:
     \begin{enumerate}
        \item if $n_i = n_j$, then $X_{i j} = g_{n_i} ( \Lambda )$,
        \item if $n_i \leq n_j$, then $X_{i j} = \left[ \begin{array}{cc} g_{n_i} ( \Lambda ) & 0 \end{array}\right] $,
        \item if $n_i \geq n_j$, then $X_{i j} = \left[ \begin{array}{c} g_{n_j} ( \Lambda ) \\ 0 \end{array}\right] $,
    \end{enumerate}
   where for each $i, j \in\{ 1, \ldots, m\} $, $g_{n_i} ( \Lambda ) $ is defined as in Theorem \ref{theorem: e J jordan block first kind complex conjugate eigenvalues}.
\end{theorem}

\noindent \textit{Proof:}

    The proof is similar to that of Theorem \ref{theorem: e J jordan block second kind real eigenvalue}.

\rightline{$\square$}

\begin{theorem}
\label{theorem: e J jordan matrix}
    Let $J$ be a generalized Jordan matrix due to Definition \ref{def:Jordan-generalized-matrix} and
    $g\in \mathbf{Sym}_{m+2n}$ a matrix function such that $g_{, \xi } = J g$. Then
    \begin{equation}
    \fl g = \diagonal{
        g_{{m^1_1} , \ldots , {m^1_{r_1}}} ( \lambda_1 ),
        \ldots,
        g_{{m^p_1} , \ldots , {m^p_{r_p}}} ( \lambda_p ),
        g_{{n^1_1} , \ldots , {n^1 {s_1}}} ( \Lambda_1 ),
        \ldots,
        g_{{n^q_1} , \ldots , {n^q {s_q}}} ( \Lambda_q )
    }
    \end{equation}
    where $g_{{m^i_1} , \ldots , {m^i_{r_i}}} ( \lambda_i )$ and $g_{{n^k_1} , \ldots , {n^k_{s_k}}} ( \Lambda_k ) $ are the functions defined as in Theorems \ref{theorem: e J jordan block second kind real eigenvalue} and \ref{theorem: e J jordan block second kind complex conjugate eigenvalues}, respectively, for each $i, j \in\{ 1, \ldots, p\} $ and $k, l\in\{ 1, \ldots, q \} $.
\end{theorem}

\noindent \textit{Proof:}

    Applying $g = g^T $ to $g_{, \xi } = J g$ we get $J g = g J^T $, then $g\in \mathcal{I} (J)$.
    By Theorem \ref{theorem: ia jordan matrix} we have $ g = \diagonal{ X _1  ( \xi ), \ldots, X _p  ( \xi ), Y _1  ( \xi ), \ldots, Y _q  ( \xi ) } $, where $ X _i ( \xi ) \in \RMset{m ^i} $ and $ Y _k ( \xi ) \in \RMset{n ^k} $ are matrix functions for $ i \in \{ 1, \ldots, p \} $ and $ k \in \{ 1, \ldots, q \} $.
    The linear differential equation $g_{, \xi } = J g$ implies $X_{i ,\xi } = J_{m^i_1, \ldots, m^i_{r_i}} (\lambda_i ) X_i$ and $Y_{k ,\xi } = J_{n^k_1, \ldots, n^k_{s_k}} (\Lambda_k ) Y_k $. By Theorems \ref{theorem: e J jordan block second kind real eigenvalue} and \ref{theorem: e J jordan block second kind complex conjugate eigenvalues} we get $ X_i = g_{{m^i_1} , \ldots , {m^i_{r_i}}} ( \lambda_i ) $ and $Y_k = g_{{n^k_1} , \ldots , {n^k_{s_k}}} ( \Lambda_k ) $, respectively, for each $i \in\{ 1, \ldots, p\} $ and $k \in\{ 1, \ldots, q\} $.

\rightline{$\square$}

\section{Equivalence classes for the matrix $A$}    

In this section we resume some facts from linear algebra which permit to describe the similarity equivalence classes for the matrix $A\in SL(n,\mathbb{R})$ from Section \ref{sec:Field_equations}, recall that $A$ is a real traceless matrix which satisfies that $Ag=gA^T$.

\begin{definition}
    A real square matrix is non-derogatory if its minimal polynomial and characteristic polynomial are equal.
\end{definition}

\begin{definition}
    Let
    \begin{equation}
        p(\lambda ) = \lambda ^n + a_{n-1} \lambda ^{n-1} + \ldots + a_1 \lambda + a_0
    \end{equation}
    be a polynomial and $a_i \in\mathbb{R}$ for $i = \{ 1, \ldots, n \} $.
    The matrix
    \begin{equation}
    \label{eq:natural normal cell}
        \left[ \begin{array}{ccccc}
        		0	&	1	 &	0	 &	\cdots	&	0        \\
	            0	&	0	 &	1	 &	\cdots	&	0        \\
	       \vdots	&	\vdots	&	\vdots	&	\ddots	&	\vdots \\
	            0	&	0	 &	0	 &	\cdots	&	1         \\
	        -a_0	&	-a_1 &	-a_2 &	\cdots	&	-a_{n-1}  \end{array}\right]
    \end{equation}
    is the companion matrix of the polynomial $ p(\lambda ) $.
    The matrices of the form (\ref{eq:natural normal cell}) are called natural normal cells.
\end{definition}

\begin{theorem}[from \cite{horn_johnson_1985}]
\label{theorem: companion matrix}
    Let $ A $ be a real square matrix with characteristic polynomial $ p ( \lambda ) $.
    If $A$ is non-derogatory, then $A$ is similar to the companion matrix of $ p ( \lambda ) $.
\end{theorem}

\begin{definition}
\label{def: natural normal matrix}
	Let $n_1, \ldots, n_m $ be positive integers such that $n = n_1 + \ldots + n_m$.
	A matrix of the form
	\begin{equation}
            A = \diagonal{ A_1, \ldots, A_m } \in \mathbf{M}_n
	\end{equation}
	is called natural normal form if
	\begin{enumerate}
		\item  $A_i \in \mathbf{M}_{n_i}$ are natural normal cell with characteristic polynomial $p_i(\lambda ) $ for $i \in \{ 1, \ldots, m \} $,
		\item  for every $j \in \{ 1, \ldots, m-1\} $, the polynomial $p_j(\lambda ) $ is a divisor of $p_{j+1}(\lambda ) $.
	\end{enumerate}
\end{definition}

\begin{theorem}[from \cite{Gantmacher59}]
\label{theorem: A natural normal form}
	Every real square matrix is similar to a unique natural normal form.
\end{theorem}

\begin{definition}
    Let
    \begin{equation}
        P = \left[ \begin{array}{ccc}
            p_{11} ( \lambda )  &   \cdots  &   p_{1n} ( \lambda )
        \\  \vdots  &   \ddots  &   \vdots
        \\  p_{n1} ( \lambda )  &   \cdots  &   p_{nn} ( \lambda )
        \end{array}\right] \in \mathbf{M}_{n}
    \end{equation}
    be a polynomial matrix and $ D _k ( \lambda ) $ the greatest common divisor of all minors of order $k$ in $P$ for $ k \in \{ 1, \ldots, n \} $.
    The invariant factors of P are defined as follows:
    \begin{equation}
   \fl  d_1(\lambda ) = D_1(\lambda ),
    d_2(\lambda ) = \frac{D_2(\lambda )}{D_1(\lambda ) }, \cdots ,
    d_r(\lambda ) = \frac{D_r(\lambda )}{D_{r-1}(\lambda )},
    d_{r+1}(\lambda ) = 0, \cdots ,
    d_n (\lambda ) = 0 .
    \end{equation}
    If all minors of order $k$ are equal to zero, then $ D _k (\lambda ) = 0 $.
\end{definition}

\begin{lemma}[from \cite{Matos:1992xn}]
\label{lemma: invariants factors companion matrix}
    Let $ A \in \RMset{n} $ be the companion matrix of the polynomial $ p ( \lambda ) $.
    The invariant factors of the matrix $A$ are equal to $ 1, \ldots, 1, p ( \lambda ) $, where the number of the 1's equals $ (n - 1) $.
\end{lemma}

\begin{lemma}[from \cite{Matos:1992xn}]
\label{lemma: invariant factors natural normal form}
	Let $A$ be the matrix of Definition \ref{def: natural normal matrix}.
	The invariant factors of the matrix $A$, are equal to $ 1, \ldots, 1, p_1 ( \lambda ), \ldots,  p_m ( \lambda ) $, where the number of the 1's is given by $ (n - m) $.
\end{lemma}

\begin{theorem}[from \cite{Gantmacher59}]
\label{theorem: similar invariant factors}
    Two real square matrices are similar if and only if they have the same invariant factors.    
\end{theorem}

\begin{definition}
    Let $n$ and $m$ be positive integers such that $ 1 < m \leq n $.
    \begin{equation}
    \fl    N _{m, n} = \{
            ( n_1, \ldots, n_m ) \in \mathbb{Z} ^m:
            0 < n_1 \leq \cdots \leq n _m,
            n = n_1 + \ldots + n _m
        \}
\end{equation}
\end{definition}

\begin{theorem}
    Let $n$ and $m$ be positive integers such that $1 < m < n$.
    The equivalence classes of the matrix $ A \in \mathfrak{sl}(n,\mathbb{R} )$ are as follows:
    \begin{equation*}
        [A] _1 = \left\{
            \left[ \begin{array}{cccccc}
            0	&	1	&	0	&	\cdots	&    0   &   0  \\	
            0	&	0	&	1	&	\cdots	&    0   &   0  \\	
            \vdots	&	\vdots	&	\vdots	&	\ddots	&	\vdots  &   \vdots  \\	
            0	&	0	&	0	&	\cdots	&    0   &   1  \\	
            -a_0	&	-a_1	&	-a_2	&	\cdots	&	-a_{n-2}  &   0  \end{array}\right] \in \mathbf{M}_{n}
        \right\}
    \end{equation*}
    and $[A] _{(n_1, \ldots, n_m)}$, which is the set of matrices $\diagonal{ A_1, \ldots, A_m }$, where the matrices $A_1, \ldots, A_m$ satisfy the following:
    
    \noindent -- $ A_i \in \mathbf{M}_{n_i} $ are natural normal cells for $ i = \{ 1, \ldots, m \} $,
    
    \noindent -- $ (n_1, \ldots, n_m) \in N_{n, m} $,
  
    \noindent -- $ p_{A _j} ( \lambda ) $ is a divisor of $p_{A_{j+1}}(\lambda )$ for $j = \{ 1, \ldots, m-1 \} $,
   
    \noindent -- $ \tr A_1 + \ldots + \tr A_m = 0 $ .
\end{theorem}
\noindent \textit{Proof:}

    Let $ X \in \mathfrak{sl} ( n, \mathbb{R} ) $.
    By the Theorems \ref{theorem: companion matrix} and \ref{theorem: A natural normal form} we have that if $X$ is non-derogatory, then $X$ is similar to a natural normal cell, or, is similar to a natural normal form.
    Suppose that $X$ is similar to $A$.

    First case, $A$ has the form (\ref{eq:natural normal cell}).
    Since that $\tr X = 0 $, then $ \tr A = 0 $, so that $ a _{n - 1} = 0 $.

    Second case, there exist an integer $ m \in \{ 2, \ldots, n \} $ such that $A$ has the form $ \diagonal{ A_1, \ldots, A_m } $, where $ A_i $ are natural normal cell with characteristic polynomial $ p _{A _i} ( \lambda ) $ of degree equal to $ n_i $ for $i = \{ 1, \ldots, m \} $ and $ n = n_1 + \ldots + n _m $.
    Since that $ p _{A _j} ( \lambda ) $ is a divisor of $ p _{A _{j + 1} } ( \lambda ) $, then $ n _j \leq n _{j + 1} $ for each $ j \in \{ 1, \ldots, m-1\} $, so that $ ( n_1, \ldots, n_m ) \in N _{m, n} $.
    Using the properties of the trace of a matrix we get $ \tr A_1 + \ldots + \tr A_m = 0 $, then for $ m = n $, we have $ A = 0 _n $.
 
    By the Theorem \ref{theorem: similar invariant factors} we find that $X$ has the same invariant factors that $A$.
    This means that the equivalence classes are determined by the invariant factors of $A$.
    Therefore, $A$ is a representation of the equivalence class where $X$ belongs.

\rightline{$\square$}

\section{Example: one-dimensional $SL(5,\mathbb{R})$-subspaces} 
\label{sec:example}

As an example to illustrate our results we will find the solutions for $g$ considering $A$ as member of the Lie algebra $\mathfrak{sl} ( 5, \mathbb{R} ) $.
For this, the following steps must be performed:
\begin{enumerate}
	\item   compute the sets  $ N_{m, n} $,
	\item	find the equivalence classes for $A$,
	\item	obtain the real Jordan forms for every equivalence classes,
	\item	determine $g$ for each real Jordan form.
\end{enumerate}
The method can be used for $ n \geq 2 $. It is easy to find the sets
\begin{eqnarray}
    N_{2, 5} & = & \{ ( 1, 4 ), ( 2, 3 ) \} \\
    N_{3, 5} & = & \{ ( 1, 1, 3 ), ( 1, 2, 2 ) \} \\
    N_{4, 5} & = & \{ ( 1, 1, 1, 2 ) \}
\end{eqnarray}
Hence, we have six equivalence classes: $ \mathfrak{A} = [A] _1 $, $ \mathfrak{B} = [A] _{( 2, 3 )} $, $ \mathfrak{C} = [A] _{( 1, 4 )} $, $ \mathfrak{D} = [A] _{( 1, 2, 2 )} $, $ \mathfrak{E} = [A] _{( 1, 1, 3 )} $ and $ \mathfrak{F} = [A] _{( 1, 1, 1, 2 )} $.

In what follows, we will explain in detail how to determine $ \mathfrak{B} $. The other five equivalence classes can be obtained in a similar way, all classes are shown in Table \ref{table: class A}.
Let $ A \in \mathfrak{B} $, then $ A $ has the form $ \diagonal{ A_1, A_2 } $, where $ A_1 \in \RMset{2} $ and $ A_2 \in \RMset{3} $ are natural normal cells.
By Lemma \ref{lemma: invariant factors natural normal form} the invariant factors of the matrix $A$ are given as $ 1, 1, 1,  p_{A_1} ( \lambda ), p_{A_2} ( \lambda ) $, where $ p_{A_1} ( \lambda ) $ and $ p_{A_2} ( \lambda ) $ are characteristic polynomials of $ A_1 $ and $ A_2 $, respectively.
Note that the degree of the polynomials $ p_{A_1} ( \lambda ) $ and $ p_{A_2} ( \lambda ) $ are 2 and 3, respectively.
Now, assume that $ p_{A_1} ( \lambda ) = \lambda ^2 - b \lambda - a $, where $ a, b \in \mathbb{R} $.
Since $ p_{A_1} ( \lambda ) $ is a divisor of $ p_{A_2} ( \lambda ) $, we can suppose, without loss of generality, that $ p_{A_2} ( \lambda ) = ( \lambda - c ) p_{A_1} ( \lambda )  $.
From the characteristic polynomial of $A$, $ p_{A} ( \lambda ) = p_{A_1} ( \lambda ) p_{A_2} ( \lambda )  $, we find $ \tr A = - 2b - c = 0 $, then $ p_{A_2} ( \lambda ) = ( \lambda + 2 b ) (  \lambda ^2 - b \lambda - a )  $.
The matrices $ A_1 $ and $ A_2 $ are also the companion matrices of $ p_{A_1} ( \lambda ) $ and $ p_{A_2} ( \lambda ) = \lambda ^3 + b \lambda ^2 - c \lambda - 2 a b $, respectively, hence
\begin{equation}  
    A_1 = \left[ \begin{array}{c c}
            0 & 1
	\\	a & b
	\end{array} \right],
    A_2 = \left[ \begin{array}{c c c}
		0	&	1	&	0
	\\	0	&	0	&	1
	\\	2 a b	&	c	&	-b
	\end{array} \right]
\end{equation}
where $ c = a + 2 b ^2 $.

In order to obtain the real Jordan forms of $ \mathcal{B} $, we consider the fact that a quadratic equation with real coefficients can have either one or two distinct real roots, or a pair of complex conjugate roots. Hence we can rewrite
\begin{enumerate}
	\item  $ p_{A_1} ( \lambda ) = ( \lambda - r _1 ) ( \lambda - r _2 ) $, $ p_{A_2} ( \lambda ) = ( \lambda - r _1 ) ( \lambda - r _2 ) ^2 $, where $ r _1 \neq r _2 $
	\item  $ p_{A_1} ( \lambda ) = ( \lambda - r _1 ) ( \lambda - r _2 ) $, $ p_{A_2} ( \lambda ) = ( \lambda - r _1 ) ( \lambda - r _2 ) ( \lambda - r _3 ) $, where $ r _1 \neq r _2 \neq r _3 $ 
	\item  $ p_{A_1} ( \lambda ) = ( \lambda - r _1 ) ^2 $, $ p_{A_2} ( \lambda ) = ( \lambda - r _1 ) ^3 $.
	\item  $ p_{A_1} ( \lambda ) = ( \lambda - r _1 ) ^2 $, $ p_{A_2} ( \lambda ) = ( \lambda - r _1 ) ^2 ( \lambda - r _2 ) $, where $ r _1 \neq r _2 $.
	\item  $ p_{A_1} ( \lambda ) = ( \lambda - r _1 ) ^2 + \theta ^2 $, $ p_{A_2} ( \lambda ) = ( ( \lambda - r _1 ) ^2 + \theta ^2 ) ( \lambda - r _2  ) $, where $ \theta > 0 $
\end{enumerate}
so that $ A_1 $ and $ A_2 $ are similar to 
\begin{enumerate}
	\item	$ \diagonal{ J _1 ( r_1 ), J _1 ( r_2 ) } $ and $ \diagonal{ J _1 ( r_1 ), J _2 ( r_2 ) } $
	\item	$ \diagonal{ J _1 ( r_1 ), J _1 ( r_2 ) } $ and $ \diagonal{ J _1 ( r_1 ), J _1 ( r_2 ), J _1 ( r_3 ) } $
	\item $ J _2 ( r_1 ) $ and $ J _3 ( r_1 ) $	
	\item $ J _2 ( r_1 ) $ and $ \diagonal{ J _2 ( r_1 ), J _1 ( r_2 ) } $
	\item	$ J _1 \left[ \begin{array}{cc}
	   r_1 & -\theta \\
	     \theta & r_1
	\end{array} \right] $ and $ \diagonal{ J _1 \left[ \begin{array}{cc}
	   r_1 & -\theta \\
	     \theta & r_1
	\end{array} \right], J _1 ( r_2 ) } $	
\end{enumerate}
respectively.  Therefore, $ A $ is similar to 
\begin{enumerate}
	\item	$ \diagonal{ J _{1,1} ( r_1 ), J _{1,2} ( r_2 ) } $
	\item	$ \diagonal{ J _{1,1} ( r_1 ), J _{1,1} ( r_2 ), J _1 ( r_3 ) } $
	\item $ J _{2,3} ( r_1 ) $	
	\item $ \diagonal{ J _{2,2} ( r_1 ), J _1 ( r_2 ) } $
	\item	$ \diagonal{ J _{1,1} \left[ \begin{array}{cc}
	   r_1 & -\theta \\
	     \theta & r_1
	\end{array} \right], J _1 ( r_2 ) } $	
\end{enumerate}
Applying the condition $ \tr A_1 + \tr A_2 = 0 $ we get
\begin{enumerate}
	\item	$ r_1 = -3 q / 2, r_2 = q, q \neq 0  $
	\item	$ 2 r_1 + 2 r_2 + r_3 = 0 $
	\item $ r_1 = 0 $	
	\item $ r_1 = q, r_2 = -4 q, q \neq 0 $
	\item	$ r_1 = q, r_2 = -4 q, q \neq 0 $	
\end{enumerate}
The real Jordan forms for every equivalence class of $A$ is presented in the Tables \ref{table: g class A}, \ref{table: g class B}, \ref{table: g class C}, \ref{table: g class D}, \ref{table: g class E} and \ref{table: g class F}.
Note that $q$ is a real constant, also $r$ and $\theta$, with or without indices, are real numbers.

Finally, we determine $g$ for $ J _{2,3 } ( 0 ) $.
By Theorem \ref{theorem: e J jordan block second kind real eigenvalue} we obtain $ g = \left[ \begin{array}{cc} g _{11} & g _{12} \\ g ^T _{12} & g _{22} \end{array} \right] $, where $ g_{12} = \left[ \begin{array}{cc} h & 0 \end{array} \right] \in \RMset{2 \times 3} $;  $ g _{11}, h \in \RMset{2} $ and $ g _{22} \in \RMset{3} $ are matrix functions given by Theorem \ref{theorem: e J jordan block first kind real eigenvalue}.
Thus
\begin{eqnarray}
    g _{11} & = & \left[ \begin{array}{cc}
		A_1 + A_2 \xi	& A_2
    \\	A_2 & 0
    \end{array} \right],
\\  g _{22} & = & \left[ \begin{array}{c c c}
			B_1 + B_2 \xi + B_3 \xi ^2 /2	&	B_2 + B_3 \xi	&	B_3
		\\	B_2 + B_3 \xi	&	B_3	&	0
		\\	B_3	&	0	&	0
		\end{array} \right],
\\  g _{12} & = & \left[ \begin{array}{c c c}
			C_1 + C_2 \xi	&	C_2	&	0
		\\	C_2	&	0	&	0
    \end{array} \right]   
\end{eqnarray}
The Tables \ref{table: g class A}, \ref{table: g class B}, \ref{table: g class C}, \ref{table: g class D}, \ref{table: g class E} and \ref{table: g class F} show the other solutions.
Note that all letters $ A, B, C, D, E $, with and without indices, are real constants.

In general relativity, the Boyer-Lindquist coordinates are very important.
They are defined as $\rho = \sqrt{ r^2 - 2 m r + \sigma^2 } \sin \theta$ and $\zeta = ( r - m ) \cos \theta$, where $m$ and $\sigma$ are constant parameters.
The Laplace equation (\ref{laplace eq}) is transform to
\begin{equation}
\label{Laplace eq Boyer-Lindquist coordinates}
    ( ( r^2 - 2 m r + \sigma^2 ) \xi _{, r} ) _{, r} + \frac{1}{\sin \theta} ( \xi _{, \theta} \sin \theta ) _{, \theta} = 0 \ ,
\end{equation}
Some solutions of (\ref{Laplace eq Boyer-Lindquist coordinates}) can be found in \cite{Matos2010}.
As an example we consider that the parameter $\xi$ depends only on $r$ and $\sigma = 0$, then
\begin{equation}
    \xi = \frac{ \gamma }{ 2 m } \ln \left( 1 - \frac{2 m}{r} \right) + \delta
\end{equation}
where $\gamma$ and $\delta$ are real constant.
For $n = 2$, we choose $A = \diagonal{ \lambda, -\lambda }$, then its corresponding matrix $g$ is $\diagonal{ \epsilon e^{\lambda \xi} , - e^{-\lambda \xi}/\epsilon }$, where $\lambda$ and $\epsilon$ are real constant. 
Thus, $g = \diagonal{ -C \left( 1 - \frac{2 m}{r} \right) ^{-p} , \left( 1 - \frac{2 m}{r} \right) ^p /C }$, where $p = - \frac{\lambda \gamma}{2m}$ and $C$ is a real constant.
Also, the differential equations for the function $f$ (\ref{diff eqs f}) are transform to
\begin{eqnarray}
    \left( \ln f \sqrt\rho \right)_{, r} & = & \frac{ 2 m^2 p^2 \sin ^2 \theta }{ r^2 - 2 m r + m^2 \sin^2 \theta } \frac{ r - m }{ r^2 - 2 m r } \\
    \left( \ln f \sqrt\rho \right)_{, \theta} & = & - \frac{ 2 m^2 p^2 \sin \theta \cos \theta }{ r^2 - 2 m r + m^2 \sin^2 \theta }    
\end{eqnarray}
Solving them, we get
\begin{equation}
    f = \frac{ D \Delta^{-p^2} }{\sqrt\rho}
\end{equation}
where $D$ is a constant and
\begin{equation}
    \Delta = 1 + \frac{m^2 \sin^2 \theta}{r^2 -2 m r}
\end{equation}
Therefore, a exact solution to EFE is
\begin{equation}
\eqalign{
    \hat g & = \frac{ D \Delta^{1 -p^2} }{\sqrt\rho} \left( dr \otimes dr + ( r^2 - 2 m r ) d\theta \otimes d\theta \right) \\
    & - \frac{\rho}{C} \left( 1 - \frac{2 m}{r} \right)^p d t \otimes d t
    + C \rho \left( 1 - \frac{2 m}{r} \right)^{-p} d x^4 \otimes d x^4    
}
\end{equation}

\begin{landscape}

\begin{longtable}{c|c|c}
    \textbf{Class}    &   \textbf{Matrices}    &   \textbf{Invariant factors}  \\ \hline\hline
    $\mathfrak{A}$   &   $\left[ \begin{array}{ccccc}
        0   &   1   &   0   &   0   &   0   \\
        0   &   0   &   1   &   0   &   0   \\
        0   &   0   &   0   &   1   &   0   \\
        0   &   0   &   0   &   0   &   1   \\
        a   &   b   &   c   &   d   &   0
    \end{array}\right]$   &    $ 1, 1, 1, 1, \lambda ^5 -d \lambda ^3 - c \lambda ^2 - b \lambda - a $   \\
    $\mathfrak{B}$   &
    $\left[ \begin{array}{ccccc}
        0  &  1  &     &     &     \\
        a  &  b  &     &     &     \\
           &     &  0  &  1  & 0   \\
           &     &  0  &  0  & 1   \\
           &     & 2ab &  c  & -b   \end{array}\right] $
     &    $ 1, 1, 1, \lambda ^2 - b \lambda - a, ( \lambda + 2 b ) ( \lambda ^2 - b \lambda - a ) $; $c = a + 2 b ^2 $   \\
    $\mathfrak{C}$   &
    $\left[ \begin{array}{ccccc}
        q   &       &       &       &       \\
            &   0   &   1   &   0   &   0   \\
            &   0   &   0   &   1   &   0   \\
            &   0   &   0   &   0   &   1   \\
            &   aq  &   c   &   d   &   -q   \end{array}\right] $
     &    $ 1, 1, 1, \lambda - q, ( \lambda - q ) ( \lambda ^3 + 2 q \lambda ^2 + b \lambda + a ) $; $ c = ( b q - a ) , d = 2 q ^2 - b $   \\
    $\mathfrak{D}$   &
    $\left[ \begin{array}{ccccc}
        q   &        &       &         &       \\
            &  0     &   1   &         &       \\
            & 3q^2/2 & -q/2  &         &       \\
            &        &       &  0      &   1   \\
            &        &       & 3q^2/2  & -q/2   \end{array}\right] $
    &    $ 1,  1, \lambda - q, ( \lambda - q ) ( \lambda + 3 q / 2 ), ( \lambda - q ) ( \lambda + 3 q / 2 ) $   \\
    $\mathfrak{E}$   &
    $\left[ \begin{array}{ccccc}
        q   &      &     &     &       \\
            &  q   &     &     &       \\
            &      &  0  &  1  &  0    \\
            &      &  0  &  0  &  1    \\
            &      &  aq &  b  &  -2q    \end{array}\right]$
    &    $ 1,  1, \lambda -q, \lambda -q, ( \lambda - q ) ( \lambda ^2 + 3 q \lambda + a ) $; $ b = 3 q ^2 - a $   \\
    $\mathfrak{F}$   &
    $\left[ \begin{array}{ccccc}
        q   &      &     &    &     \\
            &  q   &     &    &     \\
            &      &  q  &    &     \\
            &      &     & 0  & 1   \\
            &      &     & 4q & -3q  \end{array}\right]$
    &    $ 1,  \lambda - q,  \lambda - q, \lambda - q, (  \lambda - q ) (  \lambda + 4 q ) $   \\
\caption{Equivalence classes for the matrix $A \in \mathfrak{sl} (5,\mathbb{R} )$.}
\label{table: class A}
\end{longtable}

\begin{longtable}{c|c|c}
    $\mathbf{A}$  &   $\mathbf{g}$  &   \\\hline\hline
    $\left[ \begin{array}{ccccc}
    	r_1 &     &     &     &      \\
    	    & r_2 &     &     &      \\
    	    &     & r_3 &     &      \\
    	    &     &     & r_4 &      \\
    	    &     &     &     & r_5  \end{array}\right] $	
    &	$\left[ \begin{array}{ccccc}
        X_1 & 0     & 0     & 0     & 0    \\
        0   & X_2   & 0     & 0     & 0    \\
        0   & 0     & X_3   & 0     & 0    \\
        0   & 0     & 0     & X_4   & 0    \\
        0   & 0     & 0     & 0     & X_5  \end{array}\right] $	
    &   \begin{tabular}{l}
        $ X _i = A _i e ^{ r _i \xi } $ for $ i = 1, \ldots, 5 $ \\
        $ r _1 + r _2 + r _3 + r _4 + r _5 = 0 $	\\
        $ r _1 \neq r _2 \neq r _3 \neq r _4 \neq r _5 $
    \end{tabular}\\\hline
   $\left[ \begin{array}{ccccc}
    	r_1 &     &     &     &      \\
    	    & r_2 &     &     &      \\
    	    &     & r_3 &     &      \\
    	    &     &     & r_4 &  1   \\
    	    &     &     &     & r_4  \end{array}\right] $
    &	$\left[ \begin{array}{ccccc}
        X_1 & 0   & 0   & 0   & 0    \\
        0   & X_2 & 0   & 0   & 0    \\
        0   & 0   & X_3 & 0   & 0    \\
        0   & 0   & 0   & Y_1 & Y_2  \\
        0   & 0   & 0   & Y_2 & 0    \end{array}\right] $
    &   \begin{tabular}{l}
        $ X _i = A _i e ^{ r _i \xi } $ for $ i = 1, 2, 3 $ \\
        $ Y _1 = \  ( B _1 + B _2 \xi ) e ^{ r _4 \xi } $  \\
        $ Y _2 = B _2 e ^{ r _4 \xi } $ \\
        $ r _1 + r _2 + r _3 + 2  r _4 = 0 $	\\
        $ r _1 \neq r _2 \neq r _3 \neq r _4 $
    \end{tabular}\\\hline
    $\left[ \begin{array}{ccccc}
    	r_1 &     &     &     &      \\
    	    & r_2 &     &     &      \\
    	    &     & r_3 &  1  &  0   \\
    	    &     &     & r_3 &  1   \\
    	    &     &     &     & r_3  \end{array}\right] $
    &	$\left[ \begin{array}{ccccc}
        X_1 & 0   & 0   & 0   & 0   \\
        0   & X_2 & 0   & 0   & 0   \\
        0   & 0   & Y_1 & Y_2 & Y_3 \\
        0   & 0   & Y_2 & Y_3 & 0   \\
        0   & 0   & Y_3 & 0   & 0   \end{array}\right] $
    &   \begin{tabular}{l}
        $ X _i = A _i e ^{ r _i \xi } $ for $ i = 1, 2 $ \\
        $ Y _1 = \  ( B _1 + B _2 \xi + B _3 \frac{\xi ^2}{2} ) e ^{ r _3 \xi } $  \\
        $ Y _2 = \  ( B _2 + B _3 \xi ) e ^{ r _3 \xi } $  \\
        $ Y _3 = B _3 e ^{ r _3 \xi } $	\\
        $ r _1 + r _2 + 3 r _3 = 0 $	\\
        $ r _1 \neq r _2 \neq r _3 $
    \end{tabular}\\\hline
   $\left[ \begin{array}{ccccc}
    	r_1 &     &     &     &      \\
    	    & r_2 &  1  &     &      \\
    	    &     & r_2 &     &      \\
    	    &     &     & r_3 &  1   \\
    	    &     &     &     & r_3  \end{array}\right] $
    	&	$\left[ \begin{array}{ccccc}
        X & 0   & 0   & 0   & 0   \\
        0 & Y_1 & Y_2 & 0   & 0   \\
        0 & Y_2 & 0   & 0   & 0   \\
        0 & 0   & 0   & Z_1 & Z_2 \\
        0 & 0   & 0   & Z_2 & 0   \end{array}\right] $
    &   \begin{tabular}{l}
        $ X = A e ^{ r _1 \xi } $ \\
        $ Y _1 = \  ( B _1 + B _2 \xi ) e ^{ r _2 \xi } $  \\
        $ Y _2 = B _2 e ^{ r _2 \xi } $ \\
        $ Z _1 = \  ( C _1 + C _2 \xi ) e ^{ r _3 \xi } $  \\
        $ Z _2 = C _2 e ^{ r _3 \xi } $ \\
        $ r _1 + 2 r _2 + 2 r _3 = 0 $	\\
        $ r _1 \neq r _2 \neq r _3 $
    \end{tabular}\\\hline
    $\left[ \begin{array}{ccccc}
    	-4q &    &     &    &   \\
    	    & q  &	1  & 0	& 0	\\
	        &    &	q  & 1	& 0	\\
	        &    &	   & q	& 1	\\
	        &    &     &    & q  \end{array}\right] $
    &	$\left[ \begin{array}{ccccc}
        X & 0   & 0   & 0   & 0    \\
        0 & Y_1 & Y_2 & Y_3 & Y_4  \\
        0 & Y_2 & Y_3 & Y_4 & 0    \\
        0 & Y_3 & Y_4 & 0   & 0    \\
        0 & Y_4 & 0   & 0   & 0    \end{array}\right] $
    &   \begin{tabular}{l}
        $ X = A e ^{ -4 q \xi } $ \\
        $ Y _1 = \  ( B _1 + B _2 \xi + B _3 \frac{\xi ^2}{2} + B _4 \frac{\xi ^3}{3} ) e ^{ q \xi } $  \\
        $ Y _2 = \  ( B _2 + B _3 \xi + B _4 \frac{\xi ^2}{2} ) e ^{ q \xi } $  \\
        $ Y _3 = \  ( B _3 + B _4 \xi ) e ^{ q \xi } $  \\
        $ Y _4 = B _4 e ^{ q \xi } $ \\
        $ q \neq 0 $
    \end{tabular}\\\hline
    $\left[ \begin{array}{ccccc}
    	q	&	1  &  0  &                &                 \\
	        &	q  &  1  &                &                 \\
	        &      &  q  &                &                 \\
	        &      &     & -\frac{3}{2}q  & 1               \\
  	        &	   &     &                & -\frac{3}{2}q   \end{array}\right] $
    &	$\left[ \begin{array}{ccccc}
        X_1 & X_2 & X_3 & 0   & 0    \\
        X_2 & X_3 & 0   & 0   & 0    \\
        X_3 & 0   & 0   & 0   & 0    \\
        0   & 0   & 0   & Y_1 & Y_2  \\
        0   & 0   & 0   & Y_2 & 0    \end{array}\right] $
    &   \begin{tabular}{l}
        $ X _1 = \  ( A _1 + A _2 \xi + A _3 \frac{\xi ^2}{2} ) e ^{ q \xi } $  \\
        $ X _2 = \  ( A _2 + A _3 \xi ) e ^{ q \xi } $  \\
        $ X _3 = A _4 e ^{ q \xi } $ \\
        $ Y _1 = \  ( B _1 + B _2 \xi ) e^{- \frac{3 q \xi}{2}} $  \\
        $ Y _2 = B _2 e^{- \frac{3 q \xi}{2}} $ \\
        $ q \neq 0 $
    \end{tabular}\\\hline
    $\left[ \begin{array}{ccccc}
        0 & 1 & 0 & 0 & 0	\\
          &	0 &	1 &	0 &	0	\\
          &   &	0 &	1 &	0	\\
          &   &   &	0 &	1	\\
	      &   &   &   &	0   \end{array}\right] $
    &	$\left[ \begin{array}{ccccc}
        X_1 & X_2 & X_3 & X_4 & X_5   \\
        X_2 & X_3 & X_4 & X_5 & 0     \\
        X_3 & X_4 & X_5 & 0   & 0     \\
        X_4 & X_5 & 0   & 0   & 0     \\
        X_5 & 0   & 0   & 0   & 0     \end{array}\right] $
    &   \begin{tabular}{l}
        $ X _1 = A _1 + A _2 \xi + A _3 \frac{\xi ^2}{2} + A _4 \frac{\xi ^3}{6} + A _5 \frac{\xi ^4}{24} $  \\
        $ X _2 = A _2 + A _3 \xi + A _4 \frac{\xi ^2}{2} + A _5 \frac{\xi ^3}{6} $  \\
        $ X _3 = A _3 + A _4 \xi + A _5 \frac{\xi ^2}{2} $  \\
        $ X _4 = A _4 + A _5 \xi $  \\
        $ X _5 = A _5 $ \\
    \end{tabular}\\\hline
    $\left[ \begin{array}{ccccc}
    	r_1  &     &     &          &            \\
    	     & r_2 &     &          &            \\
    	     &     & r_3 &          &            \\
    	     &     &     & r_4      & -\theta_4  \\
    	     &     &     & \theta_4 &	r_4      \end{array}\right] $
    &	$\left[ \begin{array}{ccccc}
        X_1 & 0   & 0   & 0 & 0   \\
        0   & X_2 & 0   & 0 & 0   \\
        0   & 0   & X_3 & 0 & 0   \\
        0   & 0   & 0   & U & V   \\
        0   & 0   & 0   & V & -U  \end{array}\right] $
    &   \begin{tabular}{l}
        $ X _i = A _i e ^{ r _i \xi } $ for $ i = 1, 2, 3 $  \\
        $ U = e ^{ r _4 \xi } \  ( B \cos \theta _4 \xi - C \sin \theta _4 \xi )  $  \\
        $ V = e ^{ r _4 \xi } \  ( C \cos \theta _4 \xi + B \sin \theta _4 \xi )  $  \\
        $ r _1 + r _2 + r _3 + 2 r _4 = 0 $	\\
        $ r _1 \neq r _2 \neq r _3 $	\\
        $ \theta _4 > 0 $
    \end{tabular}\\\hline
    $\left[ \begin{array}{ccccc}
    	r_1   &     &     &           &             \\
    	      & r_2	& 1	  &           &             \\
    	      &	    & r_2 &           &             \\
    	      &     &     & r_3	      & -\theta_3	\\
    	      &     &     & \theta_3  &	r_3         \end{array}\right] $
    &	$\left[ \begin{array}{ccccc}
        X & 0   & 0   & 0 & 0   \\
        0 & Y_1 & Y_2 & 0 & 0   \\
        0 & Y_2 & 0   & 0 & 0   \\
        0 & 0   & 0   & U & V   \\
        0 & 0   & 0   & V & -U  \end{array}\right] $
    &   \begin{tabular}{l}
        $ X = A e ^{ r _1 \xi } $  \\
        $ Y_1 = \  ( B _1 + B _2 \xi ) e ^{ r _2 \xi } $  \\
        $ Y_2 = B _2 e ^{ r _2 \xi } $  \\
        $ U = e ^{ r _3 \xi } \  ( C \cos \theta _4 \xi - D \sin \theta _4 \xi )  $  \\
        $ V = e ^{ r _3 \xi } \  ( D \cos \theta _4 \xi + C \sin \theta _4 \xi )  $  \\
        $ r _1 + 2 r _2 + 2 r _3 = 0 $	\\
        $ r _1 \neq r_2 $	\\
        $ \theta _3 > 0 $
    \end{tabular}\\\hline
    $\left[ \begin{array}{ccccc}
        q	& 1 & 0   &                &	           \\
	        & q & 1	  &                &               \\
	        &   & q   &                &               \\
	        &   &     & -\frac{3}{2}q  & -\theta	   \\
	        &   &     &  \theta        & -\frac{3}{2}q  \end{array}\right] $
    &	$\left[ \begin{array}{ccccc}
        X_1 & X_2 & X_3 & 0 & 0    \\
        X_2 & X_3 & 0   & 0 & 0    \\
        X_3 & 0   & 0   & 0 & 0    \\
        0   & 0   & 0   & U & V    \\
        0   & 0   & 0   & V & -U   \end{array}\right] $
    &   \begin{tabular}{l}
        $ X_1 = \  ( A _1 + A _2 \xi + A _3 \frac{\xi ^2}{2} ) e ^{ q \xi } $  \\
        $ X_2 = \  ( A _2 + A _3 \xi ) e ^{ q \xi } $  \\
        $ X_3 = A _3 e ^{ q \xi } $  \\
        $ U = e^{- \frac{3 q \xi}{2}} \  ( C \cos \theta \xi - D \sin \theta \xi )  $  \\
        $ V = e^{- \frac{3 q \xi}{2}} \  ( D \cos \theta \xi + C \sin \theta \xi )  $  \\
        $ \theta > 0 $
    \end{tabular}\\\hline
    $\left[ \begin{array}{ccccc}
        q  &          &           &          &            \\
	       & r_1      & -\theta_1 &          &	          \\
	       & \theta_1 &	r_1       &          &            \\
	       &          &           & r_2	     & -\theta_2  \\
	       &          &           & \theta_2 &	r_2       \end{array}\right] $
    &	$\left[ \begin{array}{ccccc}
        X & 0   & 0    & 0   & 0     \\
        0 & U_1 & V_1  & 0   & 0     \\
        0 & V_1 & -U_1 & 0   & 0     \\
        0 & 0   & 0    & U_2 & V_2   \\
        0 & 0   & 0    & V_2 & -U_2  \end{array}\right] $
    &   \begin{tabular}{l}
        $ X = A e ^{ q \xi } $  \\
        $ U _i = e^{ r _i \xi } \  ( B _i \cos \theta _i \xi - C _i \sin \theta _i \xi )  $  \\
        $ V _i = e^{ r _i \xi } \  ( C _i \cos \theta _i \xi + B _i \sin \theta _i \xi )  $  \\
        $ q + 2 r _1 + 2 r _2 = 0 $	\\
        $ r _1 + i \theta _1 \neq r _2 + i \theta _2 $
    \end{tabular}\\\hline
     $\left[ \begin{array}{ccccc}
    	-4q    &        &         &        &          \\
	           & q	    & -\theta &	1      & 0	      \\
	           & \theta	& q	      &	0      & 1	      \\
	           &        &	      & q      & -\theta  \\
	           &        &	      & \theta & q        \end{array}\right] $
    &	 $\left[ \begin{array}{ccccc}
        X & 0   & 0    & 0   & 0    \\
        0 & U_1 & V_1  & U_2 & V_2  \\
        0 & V_1 & -U_1 & V_2 & -V_2 \\
        0 & U_2 & V_2  & 0   & 0    \\
        0 & V_2 & -U_2 & 0   & 0    \end{array}\right] $
    &   \begin{tabular}{l}
        $ X = A e ^{ -4 q \xi } $  \\
        $ U _1 = e^{ q \xi } \  ( ( B _1 + B _2 \xi ) \cos \theta \xi - ( C _1 + C _2 \xi ) \sin \theta \xi ) )  $  \\
        $ V _1 = e^{ q \xi } \  ( ( C _1 + C _2 \xi ) \cos \theta \xi + ( B _1 + B _2 \xi ) \sin \theta \xi )  $  \\
        $ U _2 = e^{ q \xi } \  ( D \cos \theta \xi - E \sin \theta \xi )  $  \\
        $ V _2 = e^{ q \xi } \  ( E \cos \theta \xi + D \sin \theta \xi )  $  \\
    \end{tabular}\\
\caption{Solutions for $g$ considering $A\in \mathfrak{A}$.}
\label{table: g class A}
\end{longtable}

\begin{longtable}{c|c|c}
    $\mathbf{A}$  &   $\mathbf{g}$    & \\\hline\hline
    $\left[ \begin{array}{ccccc}
    0   &   1 &    &      &     \\
        &   0 &    &      &     \\
        &     &  0  &  1  &  0  \\
        &     &     &  0  &  1  \\
        &     &     &     &  0  \end{array}\right] $
    &	$\left[ \begin{array}{ccccc}
        X_1 & X_2   & Z_1 & Z_2 & 0   \\
        X_1 & 0     & Z_2 & 0   & 0   \\
        Z_1 & Z_2   & Y_1 & Y_2 & Y_3 \\
        Z_2 & 0     & Y_2 & Y_3 & 0   \\
        0   & 0     & Y_3 & 0   & 0   \end{array}\right] $
    &   \begin{tabular}{l}
        $ X _1 = A _1 + A _2 \xi $  \\
        $ X _2 = A _2 $  \\
        $ Y _1 = B _1 + B _2 \xi + B _3 \frac{\xi ^2}{2} $  \\
        $ Y _2 = B _2 + B _3 \xi $  \\
        $ Y _3 = B _3 $  \\
        $ Z _1 = C _1 + C _2 \xi $  \\
        $ Z _2 = C _2 $  \\
    \end{tabular}\\\hline
    $\left[ \begin{array}{ccccc}
    r_1  &     &     &     &      \\
         & r_1 &     &     &      \\
         &     & r_2 &     &      \\
         &     &     & r_2 &      \\
         &     &     &     & r_3  \end{array}\right] $
    &	$\left[ \begin{array}{ccccc}
        X_1 & X_2 & 0   & 0   & 0   \\
        X_2 & X_3 & 0   & 0   & 0   \\
        0   & 0   & Y_1 & Y_2 & 0   \\
        0   & 0   & Y_2 & Y_3 & 0   \\
        0   & 0   & 0   & 0   & Z   \end{array}\right] $
    &   \begin{tabular}{l}
        $ X _i = A _i e^{\xi {r}_{1}} $ for $ i = 1, 2, 3 $ \\
        $ Y _j = B _i e^{\xi {r}_{2}} $ for $ j = 1, 2, 3 $ \\
        $ Z = C e ^{ -2 b \xi } $   \\
        $ 2 r _1 + 2 r _2 + r _3 = 0 $	\\
        $ r _1 \neq r _2 \neq r_3 $
    \end{tabular}\\\hline
    $\left[ \begin{array}{ccccc}
    	-4q  &    &    &   &    \\
    	     & q  & 1  &   &    \\
    	     &    & q  &   &    \\
    	     &    &    & q & 1  \\
    	     &    &    &   & q   \end{array}\right] $
    &	$\left[ \begin{array}{ccccc}
        X & 0   & 0     & 0     & 0    \\
        0 & Y_1 & Y_2   & T_1   & T_2  \\
        0 & Y_2 & 0     & T_2   & 0    \\
        0 & T_1 & T_2   & Z_1   & Z_2  \\
        0 & T_2 & 0     & Z_2   & 0    \end{array}\right] $
    &   \begin{tabular}{l}
        $ X =  A e^{- 4 q \xi} $  \\
        $ Y _1 =  \  ( B _1 + B _2 \xi ) e^{ q \xi} $  \\
        $ Y _2 =  B _2 e^{ q \xi} $  \\
        $ Z _1 =  \  ( C _1 + C _2 \xi ) e^{ q \xi} $  \\
        $ Z _2 =  C _2 e^{ q \xi} $  \\
        $ T _1 =  \  ( D _1 + D _2 \xi ) e^{ q \xi} $  \\
        $ T _2 =  D _2 e^{ q \xi} $  \\
        $ q \neq 0 $
    \end{tabular}\\\hline
    $\left[ \begin{array}{ccccc}
    q  &   &   &               &                 \\
       & q & 1 &               &                 \\
       &   & q &               &                 \\
       &   &   & -\frac{3}{2}q &                 \\
       &   &   &               & -\frac{3}{2}q   \end{array}\right] $
    &	$\left[ \begin{array}{ccccc}
        X_1 & X_2 & 0   & 0   & 0     \\
        X_2 & Y_1 & Y_2 & 0   & 0     \\
        0   & Y_2 & 0   & 0   & 0     \\
        0   & 0   & 0   & Z_1 & Z_2   \\
        0   & 0   & 0   & Z_2 & Z_3   \end{array}\right] $
    &   \begin{tabular}{l}
        $ X _i = A _i e^{q \xi} $ for $ i = 1, 2 $  \\
        $ Y _1 = \  ( B _1 + B _2 \xi ) e^{q \xi} $    \\
        $ Y _2 = B _2 e^{q \xi} $    \\
        $ Z _j = C _j e^{- \frac{3 q \xi}{2}} $ for $ j = 1, 2, 3 $    \\
        $ q \neq 0 $
    \end{tabular}\\\hline
    $\left[ \begin{array}{ccccc}
    	-4q  &        &         &        &         \\
    	     & q      & -\theta &        &         \\
   	         & \theta & q       &        &         \\
    	     &        &         & q      & -\theta \\
   	         &        &         & \theta & q       \end{array}\right] $
    &	$\left[ \begin{array}{ccccc}
        X & 0   & 0     & 0     & 0     \\
        0 & U_1 & V_1   & U_2   & V_2   \\
        0 & V_1 & -U_1  & V_2   & -U_2  \\
        0 & U_2 & V_2   & U_3   & V_3   \\
        0 & V_2 & -U_2  & V_3   & -U_3  \end{array}\right] $
    &   \begin{tabular}{l}
        $ X = A e^{- 4 q \xi} $ \\
        $ U _i = e^{ q \xi } \  ( B _i \cos \theta \xi - C _i \sin \theta \xi ) $ \\
        $ V _i = e^{ q \xi } \  ( C _i \cos \theta \xi + D _i \sin \theta \xi ) $ for $ i = 1, 2, 3 $ \\
        $ \theta \neq 0 $
    \end{tabular}\\
\caption{Solutions for $g$ considering $A\in \mathfrak{B}$.}
\label{table: g class B}
\end{longtable}

\begin{longtable}{c|c|c}
    $\mathbf{A}$  &   $\mathbf{g}$  &   \\\hline\hline
    $\left[ \begin{array}{ccccc}
    0 &   &   &   &    \\
      & 0 & 1 & 0 & 0  \\
      &   & 0 & 1 & 0  \\
      &   &   & 0 & 1  \\
      &   &   &   & 0  \end{array}\right] $
    &	 $\left[ \begin{array}{ccccc}
        A   &   B   &   0   &   0   &   0   \\
        B   &   X_1 &   X_2 &   X_3 &   X_4 \\
        0   &   X_2 &   X_3 &   X_4 &   0   \\
        0   &   X_3 &   X_4 &   0   &   0   \\
        0   &   X_4 &   0   &   0   &   0   \end{array}\right] $
    &   \begin{tabular}{l}
        $ X _1 = C _1 + C _2 \xi + C _3 \frac{\xi ^2}{2} + C _4 \frac{\xi ^3}{6} $  \\
        $ X _2 = C _2 + C _3 \xi + C _4 \frac{\xi ^2}{2} $  \\
        $ X _3 = C _3 + C _4 \xi $  \\
        $ X _4 = C _4 $  \\
    \end{tabular}\\\hline
    $\left[ \begin{array}{ccccc}
    	-4q  &   &   &   &    \\
	         & q &   &   &    \\
    	     &   & q & 1 & 0  \\
    	     &   &   & q & 1  \\
    	     &   &   &   & q   \end{array}\right] $
    &	$\left[ \begin{array}{ccccc}
        X & 0   & 0   & 0   & 0    \\
        0 & Y_1 & Y_2 & 0   & 0    \\
        0 & Y_2 & Z_1 & Z_2 & Z_3  \\
        0 & 0   & Z_2 & Z_3 & 0    \\
        0 & 0   & Z_3 & 0   & 0    \end{array}\right] $
    &   \begin{tabular}{l}
        $ X = A e^{- 4 q \xi} $    \\
        $ Y _i = B _i e^{q \xi} $ for $ i = 1, 2 $  \\
        $ Z _1 = \  ( C _1 + C _2 \xi + C _3 \frac{\xi ^2}{2} ) e^{q \xi} $ \\
        $ Z _2 = \  ( C _2 + C _3 \xi ) e^{q \xi} $ \\
        $ Z _3 = C _3 e^{q \xi} $ \\
        $ q \neq 0 $
    \end{tabular}
      \\\hline
    $\left[ \begin{array}{ccccc}
    	q & 1 & 0 &               &                \\
    	  & q & 1 &               &                \\
    	  &   & q &               &                \\
    	  &   &   & -\frac{3}{2}q &                \\
    	  &   &   &               & -\frac{3}{2}q  \end{array}\right] $
    &	$\left[ \begin{array}{ccccc}
         X_1  & X_2 & X_3 & 0   & 0     \\
         X_2  & X_3 & 0   & 0   & 0     \\
         X_3  & 0   & 0   & 0   & 0     \\
         0    & 0   & 0   & Y_1 & Y_2   \\
         0    & 0   & 0   & Y_2 & Y_3   \end{array}\right] $
    &   \begin{tabular}{l}
        $ X _1 = \  ( A _1 + A _2 \xi + A _3 \frac{\xi ^2}{2} ) e^{q \xi} $\\
        $ X _2 = \  ( A _2 + A _3 \xi ) e^{q \xi} $\\
        $ X _3 = A _3 e^{q \xi} $   \\
        $ Y _i = B _i e^{- \frac{3 q \xi}{2}} $ for $ i = 1, 2, 3 $ \\
        $ q \neq 0 $
    \end{tabular}\\\hline
   $\left[ \begin{array}{ccccc}
    q &   &   &               &                \\
      & q & 1 &               &                \\
      &   & q &               &                \\
      &   &   & -\frac{3}{2}q & 1              \\
      &	  &   &               & -\frac{3}{2}q  \end{array}\right] $
    &	$\left[ \begin{array}{ccccc}
        X_1 &   X_2 &   0   &   0   &   0   \\
        X_2 &   Y_1 &   Y_2 &   0   &   0   \\
        0   &   Y_2 &   0   &   0   &   0   \\
        0   &   0   &   0   &   Z_1 &   Z_2 \\
        0   &   0   &   0   &   Z_2 &   0  \end{array}\right] $
    &   \begin{tabular}{l}
        $ X _i = A _i e^{q \xi} $ for $ i = 1, 2 $   \\
        $ Y _1 = \  ( B _1 + B _2 \xi ) e^{q \xi} $    \\
        $ Y _2 = B _2 e^{q \xi} $    \\
        $ Z _1 = \  ( C _1 + C _2 \xi ) e^{- \frac{3 q \xi}{2}} $    \\
        $ Z _2 = C _2 e^{- \frac{3 q \xi}{2}} $    \\
        $ q \neq 0 $
    \end{tabular}\\\hline
    $\left[ \begin{array}{ccccc}
    q &   &    &     &      \\
      & q & 1  &     &      \\
      &   & q  &     &      \\
      &   &    & r_1 &      \\
      &   &    &     & r_2  \end{array}\right] $
    &	$\left[ \begin{array}{ccccc}
        X_1 &   X_2 &   0   &   0   &   0   \\
        X_2 &   Y_1 &   Y_2 &   0   &   0   \\
        0   &   Y_2 &   0   &   0   &   0   \\
        0   &   0   &   0   &   Z_1 &   0   \\
        0   &   0   &   0   &   0   &   Z_2 \end{array}\right] $
    &  \begin{tabular}{l}
        $ X _i = A _i e^{q \xi} $ for $ i = 1, 2 $   \\
        $ Y _1 = \  ( B _1 + B _2 \xi ) e^{q \xi} $    \\
        $ Y _2 = B _2 e^{q \xi} $    \\
        $ Z _i = C _i e^{ r _i \xi } $ for $ i = 1, 2 $    \\
        $ 3 q + r _1 + r _2 = 0 $	\\
        $ q \neq r _1 \neq r _2 $
    \end{tabular}\\\hline
   $\left[ \begin{array}{ccccc}
    q &   &   &               &                \\
      & q & 1 &               &                \\
      &   & q &               &                \\
      &   &   & -\frac{3}{2}q & -\theta        \\
      &   &   & \theta        & -\frac{3}{2}q  \end{array}\right] $
    &	$\left[ \begin{array}{ccccc}
        X_1 & X_2 & 0   & 0 & 0   \\
        X_2 & Y_1 & Y_2 & 0 & 0   \\
        0   & Y_2 & 0   & 0 & 0   \\
        0   & 0   & 0   & U & V   \\
        0   & 0   & 0   & V & -U  \end{array}\right] $
    &  \begin{tabular}{l}
        $ X _i = A _i e^{q \xi} $ for $ i = 1, 2 $   \\
        $ Y _1 = \  ( B _1 + B _2 \xi ) e^{q \xi} $    \\
        $ Y _2 = B _2 e^{q \xi} $    \\
        $ U = e^{- \frac{3 q \xi}{2}} \  ( C \cos \theta \xi - D \sin \theta \xi ) $    \\
        $ V = e^{- \frac{3 q \xi}{2}} \  ( D \cos \theta \xi + C \sin \theta \xi ) $    \\
        $ \theta > 0 $
    \end{tabular}\\\hline
    $\left[ \begin{array}{ccccc}
    	q &   &     &     &      \\
    	  & q &     &     &      \\
   	      &   & r_1 &     &      \\
    	  &   &     & r_2 &      \\
    	  &   &     &     & r_3  \end{array}\right] $
    &	$\left[ \begin{array}{ccccc}
        X_1 & X_2 & 0   & 0   & 0    \\
        X_2 & X_3 & 0   & 0   & 0    \\
        0   & 0   & Y_1 & 0   & 0    \\
        0   & 0   & 0   & Y_2 & 0    \\
        0   & 0   & 0   & 0   & Y_3  \end{array}\right] $
    &   \begin{tabular}{l}
        $ X _i = A _i e^{q \xi} $ for $ i = 1, 2 $   \\
        $ Y _j = B _j e^{ r _j \xi } $ for $ j = 1, 2, 3 $    \\
        $ 2 q + r _1 + r _2 + r _3 = 0 $	\\
        $ q \neq r _1 \neq r _2 \neq r _3 $
    \end{tabular}\\\hline
    $\left[ \begin{array}{ccccc}
    	q &   &     &     &       \\
    	  & q &     &     &       \\
    	  &   & r_1 &     &       \\
    	  &   &     & r_2 & 1     \\
    	  &   &     &     & r_2   \end{array}\right] $
    &	$\left[ \begin{array}{ccccc}
        X_1 & X_2 & 0 & 0   & 0   \\
        X_2 & X_3 & 0 & 0   & 0   \\
        0   & 0   & Y & 0   & 0   \\
        0   & 0   & 0 & Z_1 & Z_2 \\
        0   & 0   & 0 & Z_2 & 0   \end{array}\right] $
    &   \begin{tabular}{l}
        $ X _i = A _i e^{q \xi} $ for $ i = 1, 2 $   \\
        $ Y = B e^{ r _1 \xi } $    \\
        $ Z _1 = \  ( C _1 + C _2 \xi ) e^{ r _2 \xi } $   \\
        $ Z _2 = C _2 e^{ r _2 \xi } $   \\
        $ 2 q + r _1 + 2 r _2 = 0 $	\\
        $ q \neq r _1 \neq r _2 $
    \end{tabular}\\\hline
    $\left[ \begin{array}{ccccc}
    	q   &   &     &          &           \\
    	    & q &     &          &           \\
    	    &   & r_1 &          &           \\
    	    &   &     & r_2      & -\theta_2 \\
    	    &   &     & \theta_2 & r_2       \end{array}\right] $
    &	$\left[ \begin{array}{ccccc}
        X_1 & X_2 & 0   & 0   & 0   \\
        X_2 & X_3 & 0   & 0   & 0   \\
        0   & 0   & Y   & 0   & 0   \\
        0   & 0   & 0   & U   & V   \\
        0   & 0   & 0   & V   & -U  \end{array}\right] $
    &   \begin{tabular}{l}
        $ X _i = A _i e^{q \xi} $ for $ i = 1, 2 $   \\
        $ Y = B e^{ r _1 \xi } $    \\
        $ U = e^{ r _2 \xi } \  ( C \cos \theta \xi - D \sin \theta \xi  ) $   \\
        $ V = e^{ r _2 \xi } \  ( D \cos \theta \xi + C \sin \theta \xi  ) $   \\
        $ 2 q + r _1 + 2 r _2 = 0 $	\\
        $ q \neq r _1 $	\\
        $ \theta _2 > 0 $
    \end{tabular}\\
\caption{Solutions for $g$ considering $A\in \mathfrak{C}$.}
\label{table: g class C}
\end{longtable}

\begin{longtable}{c|c|c}
    $\mathbf{A}$  &   $\mathbf{g}$  &   \\\hline\hline
    $\left[ \begin{array}{ccccc}
    0 &   &   &   &    \\
      & 0 & 1 &   &    \\
      &   & 0 &   &    \\
      &   &   & 0 & 1  \\
      &   &   &   & 0  \end{array}\right] $
    &	$\left[ \begin{array}{ccccc}
        A_1 & A_2 & 0   & A_3 & 0   \\
        A_2 & Y_1 & Y_2 & T_1 & T_2 \\
        0   & Y_2 & 0   & T_2 & 0   \\
        A_3 & T_1 & T_2 & Z_1 & Z_2 \\
        0   & T_2 & 0   & Z_2 & 0    \end{array}\right] $
    &   \begin{tabular}{l}
        $ Y _1 = B _1 + B_2 \xi $   \\
        $ Y _2 = B _2 $   \\
        $ Z _1 = C _1 + C _2 \xi $   \\
        $ Z _2 = C _2 $   \\
        $ T _1 = D _1 + D _2 \xi $   \\
        $ T _2 = D _2 $   \\
    \end{tabular}\\\hline
    $\left[ \begin{array}{ccccc}
    q &   &   &               &                 \\
      & q &   &               &                 \\
      &   & q &               &                 \\
      &   &   & -\frac{3}{2}q &                 \\
      &   &   &               & -\frac{3}{2}q   \end{array}\right] $
    &	$\left[ \begin{array}{ccccc}
        X_1 & X_2 & X_3 & 0   & 0    \\
        X_2 & X_4 & X_5 & 0   & 0    \\
        X_3 & X_5 & X_6 & 0   & 0    \\
        0   & 0   & 0   & Y_1 & Y_2  \\
        0   & 0   & 0   & Y_2 & Y_3  \end{array}\right] $
    &   \begin{tabular}{l}
        $ X _i = A _i e^{ q \xi } $ for $ i = 1, \ldots, 6 $  \\
        $ Y _j = B _j e^{- \frac{ 3 q}{2} \xi } $ for $ j = 1, 2, 3 $   \\
        $ q \neq 0 $
    \end{tabular}\\
\caption{Solutions for $g$ considering $A\in \mathfrak{D}$.}
\label{table: g class D}
\end{longtable}

\begin{longtable}{c|c|c}
    $\mathbf{A}$  &   $\mathbf{g}$  &   \\\hline\hline
    $\left[ \begin{array}{ccccc}
    0 &   &   &   &    \\
      & 0 &   &   &    \\
      &   & 0 & 1 &	0  \\
      &   &   & 0 & 1  \\
      &   &   &   & 0  \end{array}\right] $
    &	$\left[ \begin{array}{ccccc}
        A_1 & A_2 & B_1 & 0   & 0   \\
        A_2 & A_3 & B_2 & 0   & 0   \\
        B_1 & B_2 & X_1 & X_2 & X_3 \\
        0   & 0   & X_2 & X_3 & 0   \\
        0   & 0   & X_3 & 0   & 0  \end{array}\right] $
    &   \begin{tabular}{l}
        $ X _1 = C _1 + C _2 \xi + C _3 \frac{\xi ^2}{2} $  \\
        $ X _2 = C _2 + C _3 \xi $  \\
        $ X _3 = C _3 $  \\
    \end{tabular}\\\hline
    $\left[ \begin{array}{ccccc}
    -4q  &   &   &   &    \\
         & q &   &   &    \\
         &   & q &   &    \\
         &   &   & q & 1  \\
         &   &   &   & q  \end{array}\right] $
    &	$\left[ \begin{array}{ccccc}
        X & 0   & 0   & 0   & 0     \\
        0 & Y_1 & Y_2 & Y_4 & 0     \\
        0 & Y_2 & Y_3 & Y_5 & 0     \\
        0 & Y_4 & Y_5 & Z_1 & Z_2   \\
        0 & 0   & 0   & Z_2 & 0     \end{array}\right] $
    &   \begin{tabular}{l}
        $ X = A e^{ -4 q \xi } $ \\
        $ Y _i = B _i e^{q \xi} $ for $ i = 1, \ldots, 5 $ \\
        $ Z _1 = \  ( C _1 + C _2 \xi ) e^{q \xi} $    \\
        $ Z _2 = C _2 e^{q \xi} $	\\
        $ q \neq 0 $
    \end{tabular}\\\hline
    $\left[ \begin{array}{ccccc}
    q &   &   &     &      \\
      & q &   &     &      \\
      &   & q &     &      \\
      &   &   & r_1 &      \\
      &   &   &     & r_2  \end{array}\right] $
    &	$\left[ \begin{array}{ccccc}
        X_1 & X_2 & X_3 & 0   & 0    \\
        X_2 & X_4 & X_5 & 0   & 0    \\
        X_3 & X_5 & X_6 & 0   & 0    \\
        0   & 0   & 0   & Y_1 & 0    \\
        0   & 0   & 0   & 0   & Y_2  \end{array}\right] $
    &   \begin{tabular}{l}
        $ X _i = A _i e ^{ q \xi } $ for $ i = 1, \ldots, 6 $   \\
        $ Y _j = B _j e ^{ r _j \xi } $ for $ i = 1, 2 $   \\
        $ 3 q + r _1 + r _2 = 0 $	\\
        $ q \neq r _1 \neq r _2 $
    \end{tabular}\\\hline
    $\left[ \begin{array}{ccccc}
    q &   &   &               &                \\
      & q &   &               &                \\
      &   & q &               &                \\
      &   &   & -\frac{3}{2}q & 1              \\
      &   &   &               & -\frac{3}{2}q  \end{array}\right] $
    &	$\left[ \begin{array}{ccccc}
        X_1 & X_2 & X_3 & 0   & 0     \\
        X_2 & X_4 & X_5 & 0   & 0     \\
        X_3 & X_5 & X_6 & 0   & 0     \\
        0   & 0   & 0   & Z_1 & Z_2   \\
        0   & 0   & 0   & Z_2 & 0     \end{array}\right] $
    &   \begin{tabular}{l}
        $ X _i = A _i e ^{ q \xi } $ for $ i = 1, \ldots, 6 $   \\
        $ Z _1 = \  ( C _1 + C _2 \xi ) e^{- \frac{3 q \xi}{2}} $ \\
        $ Z _2 = C _2 e^{- \frac{3 q \xi}{2}} $ \\
        $ q \neq 0 $
    \end{tabular}\\\hline
    $\left[ \begin{array}{ccccc}
    q &   &   &               &                \\
      & q &   &               &                \\
      &   & q &               &                \\
      &   &   & -\frac{3}{2}q & -\theta        \\
      &   &   & \theta        & -\frac{3}{2}q  \end{array}\right] $
    &	$\left[ \begin{array}{ccccc}
        X_1 & X_2 & X_3 & 0 & 0\\
        X_2 & X_4 & X_5 & 0 & 0\\
        X_3 & X_5 & X_6 & 0 & 0\\
        0 & 0 & 0 & U & V\\
        0 & 0 & 0 & V & -U  \end{array}\right] $
    &   \begin{tabular}{l}
        $ X _i = A _i e^{q \xi} $ for $ i = 1, \ldots, 6 $   \\
        $ U = e^{- \frac{3 q \xi}{2}} \  ( B \cos \theta \xi - C \sin \theta \xi ) $   \\
        $ V = e^{- \frac{3 q \xi}{2}} \  ( C \cos \theta \xi + B \sin \theta \xi ) $   \\
        $ \theta > 0 $
    \end{tabular}\\
\caption{Solutions for $g$ considering $A\in \mathfrak{E}$.}
\label{table: g class E}
\end{longtable}

\begin{longtable}{c|c|c}
    $\mathbf{A}$  &   $\mathbf{g}$  &   \\\hline\hline
    $\left[ \begin{array}{ccccc}
    0 &   &   &   &    \\
      & 0 &   &   &    \\
   	  &   & 0 &   &    \\
   	  &   &   & 0 & 1  \\
   	  &   &   &   & 0  \end{array}\right] $
   	 &   $\left[ \begin{array}{ccccc}
        A_1 & A_2 & A_3 & B_1 & 0   \\
        A_2 & A_4 & A_5 & B_2 & 0   \\
        A_3 & A_5 & A_6 & B_3 & 0   \\
        B_1 & B_2 & B_3 & X_1 & X_2 \\
        0   & 0   & 0   & X_2 & 0   \end{array}\right] $
    &   \begin{tabular}{l}
        $ X _1 = C _1 + C _2 \xi $\\
        $ X _2 = C _2 $\\
    \end{tabular}\\\hline
    $\left[ \begin{array}{ccccc}
    -4q  &   &   &   &    \\
         & q &   &   &    \\
         &   & q &   &    \\
         &   &   & q &    \\
         &   &   &   & q  \end{array}\right] $
    &   $\left[ \begin{array}{ccccc}
        X & 0   & 0   & 0   & 0       \\
        0 & Y_1 & Y_2 & Y_3 & Y_4     \\
        0 & Y_2 & Y_5 & Y_6 & Y_7     \\
        0 & Y_3 & Y_6 & Y_8 & Y_9     \\
        0 & Y_4 & Y_7 & Y_9 & Y_{10}  \end{array}\right] $
    &   \begin{tabular}{l}
        $ X = A e^{- 4 q \xi} $ \\
        $ Y _i = B _i e^{q \xi} $ for $ i = 1, \ldots, 10 $ \\
        $ q \neq 0 $
    \end{tabular}\\
\caption{Solutions for $g$ considering $A\in \mathfrak{F}$.}
\label{table: g class F}
\end{longtable}

\end{landscape}

\section{Conclusions}\label{sec:conclusions}

EFE are one of the most interesting and complicated equations to solve in physics. Techniques to solve them have been developed for 4-dimensions in the past. One of the most successful techniques relies on subspaces and subgroups. This method helps to generate solutions of the 4-dimensional EFE on demand, such that the Laplace equation gives the solutions for monopoles, dipoles, etc. In this work we used this technique to solve the ($n+2$)-dimensional EFE in vacuum, reducing the final matrix equation to its normal Jordan form, which permits to solve the equations with some facility. We obtained a great amount of solutions of the EFE in terms of the Laplace parameter, such that for each solution of the Laplace equations, we may get a different solution of the EFE. One can play with the different combinations of solutions to obtain even more  solutions.

\section*{Acknowledgements}

This work was partially supported by CONACyT M\'exico under grants \mbox{A1-S-8742}, 304001, 376127, 240512,  \mbox{FORDECYT-PRONACES} grant No. 490769 and \mbox{I0101/131/07 C-234/07} of the Instituto Avanzado de Cosmolog\'ia (IAC) collaboration (http://www.iac.edu.mx/).

\section*{Statements and Declarations}
I confirm that this work is original and has not been published elsewhere, nor is it currently under consideration for publication elsewhere.
On behalf of all authors, Dr. Ignacio Abraham Sarmiento Alvarado states that there is no conflict of interest.

\section*{Data availability statement}
Data sharing is not applicable to this article as no new data were created or analyzed in this study.


\section*{References}

\bibliographystyle{iopart-num}                 
\bibliography{1DsubspaceSLnR_biblio}           

\end{document}